\documentclass[draftclsnofoot]{IEEEtran}
\onecolumn
\usepackage{latexsym}
\usepackage{cite}
\usepackage{makeidx}
\makeindex
\usepackage{amssymb}
\usepackage{amsmath}
\usepackage[dvips]{graphics}
\usepackage{graphicx}
\usepackage{epsfig}
\usepackage{color}
\usepackage{psfrag}
 \allowdisplaybreaks
\usepackage{subfloat}
\usepackage{subfig}
\usepackage{setspace}
\usepackage{algorithm}
\usepackage{algpseudocode}
\usepackage{theorem}
\usepackage{textcomp}
\usepackage{helvet}
\theoremheaderfont{\upshape\slshape} \theoremstyle{plain}

\theorembodyfont{\rmfamily}

\theoremheaderfont{\itshape}

\newcommand{\ex}[1]{\mathbb{E}\left[{#1}\right]}

\newtheorem{theorem}{Theorem}

\newtheorem{lemma}{Lemma} 

\title{Extension of the Blahut-Arimoto algorithm for maximizing directed information}

\author{Iddo Naiss and Haim Permuter
\thanks{Iddo Naiss and Haim Permuter are with the Department of Electrical and Computer Engineering,
 Ben-Gurion University of the Negev, Beer-Sheva, Israel. Emails: naiss@bgu.ac.il, haimp@bgu.ac.il.
}}

\begin{document}
\maketitle
\begin{abstract} \label{Abs}
We extend the Blahut-Arimoto algorithm for maximizing Massey's directed information. The algorithm can be used for estimating the capacity of channels with delayed feedback, where the feedback is a deterministic function of the output. In order to do so, we apply the ideas from the regular Blahut-Arimoto algorithm, i.e., the alternating maximization procedure, onto our new problem. We provide both upper and lower bound sequences that converge to the optimum value. Our main insight in this paper is that in order to find the maximum of the directed information over causal conditioning probability mass function (PMF), one can use a backward index time maximization combined with the alternating maximization procedure. We give a detailed description of the algorithm, its complexity, the memory needed, and several numerical examples.
\end{abstract}
\begin{keywords}
Alternating maximization procedure, Backwards index time maximization, Blahut-Arimoto algorithm,
Causal conditioning, Channels with feedback, Directed information, Finite state channels, Ising
Channel, Trapdoor channel.
\end{keywords}
\begin{section}{Introduction} \label{secIntro}
In his seminal work, Shannon \cite{Shannon} showed that the capacity of a memoryless channel is given as the optimization problem
\begin{align}
C=\max_{p(x)}{I(X;Y)}.\label{shancap}
\end{align}
Since the set of all $p(x)$ is not of finite cardinality, an optimization method is required to find the capacity $C$.
In order to obtain an efficient way to calculate the global maximum in (\ref{shancap}), the well-known Blahut-Arimoto algorithm (referred to as BAA) was introduced by Blahut \cite{Bla} and Arimoto \cite{Ari} in 1972. The main idea is that we can calculate the optimum value using the equality $$\max_{p(x)}I(X;Y)=\max_{p(x),p(x|y)}I(X;Y),$$ i.e., we can maximize over $p(x)$ and $p(x|y)$, instead of just $p(x)$ alone. The maximization is then achieved using the alternating maximization procedure. The convergence of the alternating maximization procedure to the global maximum was proven by Csiszar and Tusnady \cite{CsiszarTusnady}, and later by Yeung \cite{Reymund}.

In this paper, we find an efficient way to estimate the capacity of channels with feedback. It was shown by Massey \cite{Massey}, Kramer \cite{Kramer1}, Tatikonda and Mitter\cite{TatiMit}, Permuter, Weissman, and Goldsmith \cite{PermuterWeissmanGoldsmith}, and Kim \cite{Kim}, that the expression
\begin{align}
C_n=\frac{1}{n}\max_{p(x^n||y^{n-1})}I(X^n\rightarrow Y^n)\nonumber
\end{align}
has an important role in characterizing the feedback capacity, where
\begin{align}
I(X^n\rightarrow Y^n)&=\sum_{y^n,x^n}{p(y^n,x^n)\log{\frac{p(y^n||x^n)}{p(y^n)}}}\nonumber
\end{align}
is the \textit{directed information}, and $p(y^n||x^n)$ is a \textit{causally conditioned} PMF (definitions in Section \ref{secPre}) given by
\begin{align}
p(y^n||x^n)=\prod_{i=1}^n{p(y_i|y^{i-1},x^i)}.\label{causal}
\end{align}
Since in the maximization we deal with causally conditioned PMFs, trying to follow the regular BAA will result in difficulties. This is due to the fact that a causal conditioned PMF is the result of multiplications of conditioned PMFs as seen in (\ref{causal}). While in the regular BAA we maximize over $p(x^n)$, and thus the constraints are simply $\sum_{x^n}p(x^n)=1$ and $p(x^n)\geq0$, in our extended problem we have no efficient way of writing all the constraints necessary for a causally conditioned PMF. In fact, we need $n$ simple constraints, one for each product of $p(x^n||y^{n-1})$.
Another difficulty is that although the equality
\begin{align}
I(X^n\rightarrow Y^n)=\sum_{i=1}^{n}I(X_i;Y_i^n|X^{i-1},Y^{i-1})\nonumber
\end{align}
holds, we cannot translate the given problem into
\begin{align}
\sum_{i=1}^{n}\max_{p(x_i|x^{i-1},y^{i-1})}I(X_i;Y_i^n|X^{i-1},Y^{i-1})\nonumber
\end{align}
since $p(x_i|x^{i-1},y^{i-1})$ influence all terms $\{I(X_j;Y_j^n|X^{j-1},Y^{j-1})\}_{j=i}^n$. A
solution could be to maximize backwards from $i=n$ to $i=1$ over $p(x_i|x^{i-1},y^{i-1})$, and it
can be shown that in each maximization, the non-causal probability $p(x_i|x^{i-1},y^n)$ is
determined only by the previous $p(x_j|x^{j-1},y^{j-1})$ for $j\geq i$. In our solution, we
maximize the entire expression $I(X^n\rightarrow Y^n)$ as a function of
$\{p(x_1),p(x_2|x_1,y_1),...,p(x_n|x^{n-1},y^{n-1}),p(x^n|y^n)\}$. Each time we maximize over a
specific $p(x_i|x^{i-1},y^{i-1})$ starting from $i=n$ and moving backwards to $i=1$, where all
but $p(x_i|x^{i-1},y^{i-1})$ are fixed.

Before we present the extension of the BAA to the directed information, let us present some of the other extensions of this algorithm. In 2004, Matz and Duhamel\cite{Matz} proposed two Blahut-Arimoto-type algorithms that often converge significantly faster than the standard Blahut-Arimoto algorithm, which relied on following the natural gradient rather than maximizing per variable. During that year, Rezaeian and Grant \cite{RezGra} generalized the regular BAA for multiple access channels, and Dupuis, Yu, and Willems extended the BAA for channels with side information \cite{DupYuWil}. They used the fact that the input is a deterministic function of the auxiliary variable and the side information, and then extended the input alphabet. Another solution to the side information problem was given by El Gamal and Heegard\cite{GamHee}, where they did not expand the alphabet, but included an additional step to optimize over $p(x|u,s)$. Also, the BAA was used by Egorov, Markavian, and Pickavance \cite{EgoMarkPick} to decode Reed Solomon codes. In 2005 Dauwels \cite{Dauw} showed how the BAA can be used to calculate the capacity of continuous channels. Dauwels's main idea is the use of sequential Monte-Carlo integration methods known as the "particle filters". In 2008 Vontobel, Kav\u{c}i\'{c}, Arnold, and Loeliger\cite{pascal} extended the regular BAA to estimate the capacity of finite state channels where the input is Markovian.
Sumszyk and Steinberg \cite{SumStein} gave a single letter characterization of the capacity of an information embedding channel and provided a BA-type algorithm for the case where the channel is independent of the host given the input.

Recently, few papers about the maximization of the directed information using control theory and
dynamic programming were published. In \cite{YangKavTati}, Yang, Kavcic and Tatikonda maximized
the directed information to estimate the feedback capacity of finite-state machine channels where
the state is a deterministic function of the previous state and input. Chen and Berger
\cite{Chen05} maximized the directed information for the case where the state of the channel is
known to the encoder and decoder in addition to the feedback link. Later, Permuter, Cuff, Van Roy
and Weissman \cite{PermuterPaulBenTsachy} maximized the directed information and found the
capacity of the trapdoor channel with feedback. In \cite{GorCol}, Gorantla and Coleman estimated
the maximum of directed information where they considered a dynamical system, whose state is an
input to a memoryless channel. The state of the dynamical system is affected by its past, an
exogenous input, and causal feedback from the channel's output.

The remainder of the paper is organized as follows. In Section \ref{secPre} we present the
notations we use throughout the paper, and give the outline for the alternating maximization
procedure as given by Yeung\cite{Reymund}. In Section \ref{secDescription} we give a description
of the algorithm for solving the optimization problem- $\max_{p(x^n||y^{n-1})}I(X^n\rightarrow
Y^n)$, calculate the complexity of the algorithm and memory needed, and compare it with those of
the regular BAA. In Section \ref{secDerivate} we derive the algorithm using the alternating
maximization procedure, and show the convergence of our algorithm to the optimum value. Numerical
examples for channel capacity with feedback are presented in Section \ref{secEx}. In Appendix
\ref{appfd} we give a wider angle on the feedback channel problem, where the feedback of the
channel is a deterministic function $f$ of the output with some delay $d$; namely, we derive the
algorithm for the optimization problem $\max_{p(x^n||z^{n-d})}I(X^n\rightarrow Y^n),$ where
$z_i=f(y_i)$ and $d\geq 1$. In Appendix \ref{proofThUp} we prove an upper bound for
$\max_{p(x^n||y^{n-d})}I(X^n\rightarrow Y^n),$ which converges to the directed information from
above and helps determining the stoping iteration of the algorithm.
\end{section}

\begin{section}{Preliminaries} \label{secPre}
\begin{subsection}{Directed information and causal conditioning}
In this section we present the definitions of directed information and causally conditioned PMF,
originally introduced by Massey\cite{Massey} (who was inspired by Marko's work \cite{Marko73} on
Bidirectional Communication) and by Kramer \cite{Kramer1}. These definitions are necessary in
order to address channels with memory. We denote by $X_1^n$ the vector $(X_1,X_2,...X_n)$.
Usually we use the notation $X^n=X_1^n$ for short. Further, when writing a PMF we simply write
$P_X(X=x)=p(x)$. Let us denote as $p(x^n||y^{n-d})$ the probability mass function (PMF) of $X^n$
\textit{causally conditioned} on $Y^{n-d}$, given by
\begin{align}
p(x^n||y^{n-d})\triangleq\prod^{n}_{i=1}{p(x_i|x^{i-1}y^{i-d})}.\label{causalpmf}
\end{align}
Here we have to establish that when $d>n$, the vector $X^{n-d}=\emptyset$.
Two straight forward properties of the causal conditioning PMF that we use throughout the paper are
\begin{align}
\sum_{x_n}p(x^n||y^{n-d})=p(x^{n-1}||y^{n-d-1}),\label{casprop1}
\end{align}
and
\begin{align}
p(x_i|x^{i-1}y^{i-d})=\frac{p(x^i||y^{i-d})}{p(x^{i-1}||y^{i-d-1})}\label{casprop2}.
\end{align}
Another elementary property is the chain rule for directed information
\begin{align}
p(x^n||y^{n-1})p(y^n||x^n)=p(x^n,y^n).
\end{align}
The definitions above lead to the causally conditioned entropy $H(X^n||Y^n)$, which is given by $$H(X^n||Y^n)\triangleq-\ex{\log p(X^n||Y^n)}.$$
Moreover, the directed information from $X^n$ to $Y^n$ is given by
\begin{align}
I(X^n\rightarrow Y^n)\triangleq H(Y^n)-H(Y^n||X^n).\label{directed1}
\end{align}
It is possible to show, that we can write the directed information as such:
\begin{align}
I(X^n\rightarrow Y^n)=\sum_{y^n,x^n}{p(y^n||x^n)r(x^n||y^{n-1})\log{\frac{q(x^n|y^n)}{r(x^n||y^{n-1})}}}.\nonumber
\end{align}
We refer to this form when using the alternating maximization procedure since $\{\textbf{r}=r(x^n||y^{n-1}),\ \textbf{q}=q(x^n|y^n)\}$ are the variables we optimize over where $p(y^n||x^n)$ is fixed. For convenience, we use from now on the notation of
\begin{align}
I(X^n\rightarrow Y^n)=\cal{I}(\textbf{r},\textbf{q})\label{calI}
\end{align}
when required.
With these definitions, we follow the alternating maximization procedure given by Yeung\cite{Reymund} in order to maximize the directed information.
\end{subsection}

\begin{subsection}{Alternating maximization procedure}
Here, we present the alternating maximization procedure on which our algorithm is based.
Let $f(u_1,u_2)$ be a real function, and let us consider the optimization problem given by
$$\sup_{u_1\in A_1, u_2\in A_2}f(u_1,u_2)=f^*.$$
We denote by $c_2(u_1)\in A_2$ the point that achieves $\sup_{u_2\in A_2}f(u_1,u_2)$, and by $c_1(u_2)\in A_1$ the one that achieves $\sup_{u_1\in A_1}f(u_1,u_2)$.
The algorithm is defined by iterations, where in each iteration we maximize over one of the variables.
Let $(u_1^0,u_2^0)$ be an arbitrary point in $A_1\times A_2$. For $k\geq0$ let
$$(u_1^k,u_2^k)=(c_1(u_2^{k-1}),c_2(c_1(u_2^{k-1}))),$$ and let $f^k=f(u_1^k,u_2^k)$ be the value if the present iteration.
The following lemma describes the conditions the problem needs to meet in order for $f^k$ to converge to $f^*$ as $k$ goes to infinity.
\begin{lemma}[Lemmas 9.4, 9.5 in \cite{Reymund}, Convergence of the alternating maximization procedure]\label{lemconv}.
Let $f(u_1,u_2)$ be a real, concave, bounded from above function that is continuous and has continuous partial derivatives, and let the sets $A_1,A_2,$ which we maximize over, be convex.
Further, assume that $c_2(u_1)\in A_2$ and $c_1(u_2)\in A_1$ for all $u_1\in A_1,\ u_2\in A_2$.
Under these conditions, $\lim_{k\rightarrow\infty}f^k=f^{*}$.
\end{lemma}

In Section \ref{secDescription} we give a detailed description of the algorithm that computes $\max_{p(x^n||y^{n-1})}I(X^n\rightarrow Y^n)$ based on the alternating maximization procedure. In Section \ref{secDerivate} we show that the conditions in Lemma \ref{lemconv} hold, and therefore the algorithm we suggest, which is based on the alternating maximization procedure, converges to the global optimum.
\end{subsection}
\end{section}

\begin{section}{Description of the algorithm} \label{secDescription}
In this section, we describe an algorithm for maximizing the directed information. In addition, we compute the complexity of the algorithm per iteration, and compare it to the complexity of the regular BAA. The memory calculation is also given.
\begin{subsection}{The algorithm for channel with feedback}
In Algorithm \ref{algc}, we present the steps required to maximize the directed information where the channel $p(y^n||x^n)$ is fixed and the delay is $d=1$.
\begin{algorithm}
\caption{Iterative algorithm for calculating $\max_{p(x^n||y^{n-1})}I(X^n\rightarrow Y^n)$, where $p(y^n||x^n)$ is fixed.}\label{algc}
\begin{itemize}
\item[(a)] Start from a random point $q(x^n|y^n)$. Usually we start from a uniform distribution, i.e., $q(x^n|y^n)=2^{-n}$ for every $(x^n,y^n)$\\
\item[(b)] Starting from $i=n$, calculate $r(x_i|x^{i-1},y^{i-1})$ using the formula\\
\begin{equation}
r(x_i|x^{i-1},y^{i-1})=\frac{r'(x^i,y^{i-1})}{\sum_{x_i}r'(x^i,y^{i-1})},\label{RIgen}
\end{equation}
where
\begin{equation}
r'(x^i,y^{i-1})=\prod_{x_{i+1}^n,y_{i}^n}{\left[\frac{q(x^n|y^n)}{\prod_{j=i+1}^{n}{r(x_j|x^{j-1},y^{j-1})}}\right]^{p(y_i|x^i,y^{i-1})\prod_{j=i+1}^{n}{r(x_j|x^{j-1},y^{j-1})p(y_j|x^j,y^{j-1})}}},\label{RIform}
\end{equation}
and do so backwards until $i=1$.\\
\item[(c)] Once you have $r(x_i|x^{i-1},y^{i-1})$ for all $i\in\{1,...,n\}$, compute $r(x^n||y^{n-1})=\prod_{i=1}^n{r(x_i|x^{i-1},y^{i-1})}$.\\
\item[(d)] Compute $q(x^n|y^n)$ using the formula
\begin{equation}
q(x^n|y^n)=\frac{r(x^n||y^{n-1})p(y^n||x^n)}{\sum_{x^n}{r(x^n||y^{n-1})p(y^n||x^n)}}.\label{Qform}
\end{equation}
\item[(e)] Calculate $I_U-I_L$, where
\begin{align}
I_L&=\frac{1}{n}\sum_{y^n,x^n}{p(y^n||x^n)r(x^n||y^{n-1})\log{\frac{q(x^n|y^n)}{r(x^n||y^{n-1})}}},\nonumber\\
I_U&=\frac{1}{n}\max_{x_1}\sum_{y_1}\max_{x_{2}}\cdots\sum_{y_{n-1}}\max_{x_n}\sum_{y_{n}}p(y^n||x^n)\log\frac{p(y^n||x^n)}{\sum_{x'^n}p(y^n||x'^n)\cdot r(x'^n||y^{n-1})}.\nonumber
\end{align}
\item[(f)] Return to (b) if $(I_U-I_L)\geq\epsilon$.
\item[(g)] $C_n=I_L$.
\end{itemize}
\end{algorithm}
Note that the regular BAA has a structure similar to that of Algorithm \ref{algc}, where step (b) is an additional backward loop. Its purpose is to maximize over the input causal probability, which is not necessary in the regular BAA.

Now, let us present a special case and a few extensions for Alg. \ref{algc}.
\begin{itemize}
\item[(1)] \textit{Regular BAA, i.e., $n=1$}. For $n=1$, the algorithm suggested here agrees with the original BAA, where instead of steps (b), (c) we have
\begin{align}
r(x)=\frac{\prod_y{q(x|y)^{p(y|x)}}}{\sum_{x}\prod_{y}q(x|y)^{p(y|x)}},\label{BA1r}
\end{align}
and step (d) is replaced by
\begin{align}
q(x|y)=\frac{r(x)p(y|x)}{\sum_{x}{r(x)p(y|x)}}.\label{BA1q}
\end{align}
The bounds $I_L,\ I_U$ agree with the regular BAA as well, and are of the form
\begin{align}
I_L&=\sum_{y,x}{p(y|x)r(x)\log{\frac{q(x|y)}{r(x)}}},\nonumber\\
I_U&=\max_{x}\sum_{y}p(y|x)\log\frac{p(y|x)}{\sum_{x'}p(y|x')\cdot r(x')}.\nonumber
\end{align}

\item[(2)] \textit{Feedback with general delay $d$}. We can generalize the algorithm in order to compute $\max_{r(x^n||y^{n-d})}I(X^n\rightarrow Y^n)$, where the feedback is the output with delay $d$. In that case, in step (b) we have
\begin{align}
r'(x^i,y^{i-d})=\prod_{x_{i+1}^n,y_{i-d+1}^n}{\left[\frac{q(x^n|y^n)}
{\prod_{j=i+1}^{n}{r(x_j|x^{j-1},y^{j-d})}}\right]^{\prod_{j=i-d+1}^{n}{p(y_j|x^{j},y^{j-1})}\prod_{j=i+1}^{n}{r(x_j|x^{j-1},y^{j-d})}}},\label{Rform1}
\end{align}
and step (d) will be replaced by
\begin{align}
q(x^n|y^n)=\frac{r(x^n||y^{n-d})p(y^n||x^n)}{\sum_{x^n}{r(x^n||y^{n-d})p(y^n||x^n)}}.\label{Qform1}
\end{align}
The bounds $I_L,\ I_U$ are of the form
\begin{align}
I_L&=\frac{1}{n}\sum_{y^n,x^n}{p(y^n||x^n)r(x^n||y^{n-d})\log{\frac{q(x^n|y^n)}{r(x^n||y^{n-d})}}},\nonumber\\
I_U&=\frac{1}{n}\max_{x^d}\sum_{y_1}\max_{x_{d+1}}\cdots\sum_{y_{n-d}}\max_{x_n}\sum_{y_{n-d+1}^n}p(y^n||x^n)\log\frac{p(y^n||x^n)}{\sum_{x'^n}p(y^n||x'^n)\cdot r(x'^n||y^{n-d})}.\nonumber
\end{align}
\item[(3)] \textit{Feedback as a function of the output with general delay}. In Appendix \ref{appfd}, we generalize the algorithm in order to compute $\max_{r(x^n||z^{n-d})}I(X^n\rightarrow Y^n)$, where the feedback $z^{n-d}$ is a deterministic function of the delayed output. The expression characterizes the capacity of channels with time-invariant feedback\cite{PermuterWeissmanGoldsmith}.
     In that case, in step (b) we have
\begin{align}
r'(x^i,z^{i-d})=\prod_{x_{i+1}^n,y_{i-d+1}^n}\prod_{A_{i,d,z}}{\left[\frac{q(x^n|y^n)}
{\prod_{j=i+1}^{n}{r(x_j|x^{j-1},z^{j-d})}}\right]^{\frac{p(y^n||x^n)\prod_{j=i+1}^{n}{r(x_j|x^{j-1},z^{j-d})}}{\sum_{A_{i,d,z}}\prod_{j=1}^{i-d}{p(y_j|x^{j},y^{j-1})}}}},\label{Rformfd}
\end{align}
where we define the set $A_{i,d,z}\triangleq\{y^{i-d}:z^{i-d}=f(y^{i-d})\}$ as the set of output sequences that $f$ transforms to $z^{i-d}$,
and step (d) will be replaced by
\begin{align}
q(x^n|y^n)=\frac{r(x^n||z^{n-d})p(y^n||x^n)}{\sum_{x^n}{r(x^n||z^{n-d})p(y^n||x^n)}}.\label{Qformfd}
\end{align}
The bounds $I_L,\ I_U$ are of the form
\begin{align}
I_L&=\frac{1}{n}\sum_{y^n,x^n}{p(y^n||x^n)r(x^n||z^{n-d})\log{\frac{q(x^n|y^n)}{r(x^n||z^{n-d})}}},\nonumber\\
I_U&=\frac{1}{n}\max_{x^d}\sum_{z_1}\max_{x_{d+1}}\cdots\sum_{z_{n-d}}\max_{x_n}\sum_{A_{n,d,z}}\sum_{y_{n-d+1}^n}p(y^n||x^n)\log\frac{p(y^n||x^n)}{\sum_{x'^n}p(y^n||x'^n)\cdot r(x'^n||z^{n-d})}.\nonumber
\end{align}
\end{itemize}

Note, that for $d=n$, the vector $z^{n-d}=\emptyset$, hence $r(x_i|x^{i-1},z^{i-d})=r(x_i|x^{i-1})$, and
$$r(x^n||z^{n-d})=\prod_{i=1}^n r(x_i|x^{i-1})=r(x^n).$$
Also note that when $f(y)=const$, $r(x^n||z^{n-d})=r(x^n)$, $A_{i,d,z}=y^{i-d}$, and
$\sum_{y^{i-d}}\prod_{j=1}^{i-d}{p(y_j|x^{j},y^{j-1})}=1$.
In each of the cases above ($d=n$ or $f(y)=const.$), in step (d) we have
$$q(x^n|y^n)=\frac{r(x^n)p(y^n||x^n)}{\sum_{x^n}{r(x^n)p(y^n||x^n)}},$$
and we obtain a different version of the regular BAA for channel capacity, where the maximization is done over all $r(x_i|x^{i-1})$ instead of over $r(x^n)$ at once. Furthermore, if $f(y)=y$ then case (3) agrees with all the equations of case (2).
\end{subsection}

\begin{subsection}{Complexity and Memory needed}
Here, we give an expression for the computation complexity of one iteration in the algorithm, and then compare it to regular BAA.
This will be done in two parts, one for each step in the iteration.
\begin{itemize}
\item[(1)] Complexity of computing $q(x^n|y^n)$ as given in (\ref{Qform}).
For each $y^n$, we need ${|\cal{X}|}^n$ multiplications for a specific $x^n$ and use the denominator computed for every other $x^n$, thus obtaining $O({|\cal{X}|}^n)$ operations. Doing so for all $y^n$ achieves $O({|\cal{X}|}^n{|\cal{Y}|}^n)=O({(|\cal{X}||\cal{Y}|)}^n)$.
\item[(2)] Complexity of computing $r(x^n||y^{n-1})$.
First, we compute the complexity of each $r(x_i|x^{i-1},y^{i-1})$ as given in (\ref{RIform}), assuming that an exponent is a constant number of computations, i.e., $O(1)$. Simple computations will conclude that the entire numerator takes about $O((n-i){(|\cal{X}||\cal{Y}|)}^{n-i})$ computations. The denominator is a summation over ${|\cal{X}|}^i$ variables, and as with $q(x^n|y^n)$, we can use the denominator for every other $x^i$. Hence, we obtain $O((n-i){(|\cal{X}||\cal{Y}|)}^n)$ computations for every $i\in\{1..n\}$. Summing over $i$ will achieve $O((n+n^2){(|\cal{X}||\cal{Y}|)}^n)=O(n^2{(|\cal{X}||\cal{Y}|)}^n)$ computations. Multiplying all $r(x_i|x^{i-1},y^{i-1})$s is a constant number of computations for every $(x_i,y_i)$. Finally, in order to compute $r(x^n||y^{n-1})$ we need $O((n^2+n){(|\cal{X}||\cal{Y}|)}^{n})$ computations.
\end{itemize}
To conclude, each iteration requires about $O(n^2{(|\cal{X}||\cal{Y}|)}^{n})$ computations.

Comparing to regular BAA: Since BAA computes the capacity of memoryless channels, we only need to compute $r(x)$ and $q(x|y)$. In much the same way, we can have its complexity and achieve  $O({(|\cal{X}||\cal{Y}|)})$ computations.
However, if we want to compare it to BAA for channels with memory, we replace $X\Leftrightarrow X^n$, $Y\Leftrightarrow Y^n$ But, $|{\cal{X}}^{n}|={|{\cal{X}}|}^n$ and so we obtain $O({(|\cal{X}||\cal{Y}|)}^{n})$ computations.
The memory needed for the algorithm is very much dependent on the manner in which one implements the algorithm. However, the obligatory memory needed is for $q,\ p$, and $r$ and its products; thus we need at least $n{(|\cal{X}||\cal{Y}|)}^n$ cells of type double.
Computation complexity and memory needed are presented in Table \ref{ComplexMemory}.
\begin{table}[h!]
\caption{Memory and operations needed for regular and extended BAA for channel coding with feedback.} 
\centering
\begin{tabular}{|c ||c | c|}
\hline
& Operation & Memory\\
\hline
&&\\
$\max_{p(x)}\left(\frac{1}{n}I(X^n;Y^n)\right)$, regular BAA for channel capacity & $O({(|\cal{X}||\cal{Y}|)}^{n})$ & ${(|\cal{X}||\cal{Y}|)}^{n}$\\
&&\\
\hline
&&\\
$\max_{p(x^n||y^{n-1})}\left(\frac{1}{n}I(X^n\rightarrow Y^n)\right)$, Alg. \ref{algc} & $O(n^2{(|\cal{X}||\cal{Y}|)}^{n})$ & $n{(|\cal{X}||\cal{Y}|)}^n$\\
&&\\
\hline
\end{tabular}
\label{ComplexMemory}
\end{table}
\end{subsection}
\end{section}

\begin{section}{Derivation of Algorithm \ref{algc}} \label{secDerivate}
In this section, we derive Algorithm \ref{algc} using the alternating maximization procedure, and conclude its convergence to the global optimum using Lemma \ref{lemconv}. Throughout the paper, note that the channel $p(y^n||x^n)$ is fixed in all maximization calculations. For this purpose we present several lemmas that will assist in proving our main goal: an algorithm for calculating $\max I(X^n\rightarrow Y^n)$. In Lemma \ref{isconcave} we show that the directed information function has the properties required for lemma \ref{lemconv}. In Lemma \ref{whyalt} we show that we are allowed to maximize the directed information over $r(x^n||y^{n-1})$ and $q(x^n|y^n)$ combined, rather than just over $r(x^n||y^{n-1})$, thus creating an opportunity to use the alternating maximization procedure for achieving the optimum value. Lemma \ref{rfix} is a supplementary claim that helps us prove Lemma \ref{whyalt}, in which we find an expression for $q(x^n|y^n)$ that maximizes the directed information where $r(x^n||y^{n-1})$ is fixed. In Lemma \ref{qfix} we find an explicit expression for $r(x^n||y^{n-1})$ that maximizes the directed information where $q(x^n|y^n)$ is fixed. Theorem \ref{Thc} combines all lemmas to show that the alternating maximization procedure as described by $I_L$ in Alg. \ref{algc} exists and converges. We end with Theorem \ref{Thupbound} that proves the existence of the upper bound, $I_U$.
\begin{lemma}\label{isconcave}.
For a fixed channel $p(y^n||x^n)$, the directed information given by
\begin{align}
I(X^n\rightarrow Y^n)=\sum_{y^n,x^n}{p(y^n||x^n)r(x^n||y^{n-1})\log{\frac{q(x^n|y^n)}{r(x^n||y^{n-1})}}}\label{rqfbform}
\end{align}
as a function of $\{\textbf{r}=r(x^n||y^{n-1}),\ \textbf{q}=q(x^n|y^n)\}$ is concave, continuous and has continuous partial derivatives.
\begin{proof}
First we need to show that the directed information can be written as above by using the causal conditioning chain rule.
\begin{align}
I(X^n\rightarrow Y^n)&=\sum_{y^n,x^n}{p(y^n,x^n)\log{\frac{p(y^n||x^n)}{p(y^n)}}}\nonumber\\
&=\sum_{y^n,x^n}p(y^n||x^n)r(x^n||y^{n-1})\log{\frac{p(y^n||x^n)r(x^n||y^{n-1})}{p(y^n)r(x^n||y^{n-1})}}\nonumber\\
&=\sum_{y^n,x^n}{p(y^n||x^n)r(x^n||y^{n-1})\log{\frac{q(x^n|y^n)}{r(x^n||y^{n-1})}}}.\nonumber
\end{align}
Then we recall the log-sum inequality \cite[Theorem 2.7.1]{CoverTomas} given by
\begin{align}
\sum_{i=1}^n{a_i\log\frac{a_i}{b_i}}\geq\left(\sum_{i=1}^n{a_i}\right)\log\frac{\sum_{i=1}^n{a_i}}{\sum_{i=1}^n{b_i}}.\label{logsum}
\end{align}
We define the sets
\begin{align}
A_1&=\{r(x^n||y^{n-1}): r(x^n||y^{n-1})>0 \text{\ is a causally conditioned PMF}\},\nonumber\\
A_2&=\{q(x^n|y^n): q(x^n|y^n) \text{\ is a conditioned PMF}\},\label{sets}
\end{align}
as the sets over which we maximize.
Now, for $(\textbf{r}_1,\textbf{q}_1),\ (\textbf{r}_2,\textbf{q}_2)$ in $A=A_1\times A_2$ and $\lambda\in[0,1]$, by using the log-sum inequality given above we derive that
\begin{eqnarray}
(\lambda r_1+(1-\lambda)r_2)\log\frac{\lambda r_1+(1-\lambda)r_2}{\lambda q_1+(1-\lambda)q_2}\nonumber
\leq\lambda r_1\log\frac{r_1}{q_1}+(1-\lambda)r_2\log\frac{r_2}{q_2}.\nonumber
\end{eqnarray}
Taking the reciprocal of the logarithms yields
\begin{eqnarray}
(\lambda r_1+(1-\lambda)r_2)\log\frac{\lambda q_1+(1-\lambda)q_2}{\lambda r_1+(1-\lambda)r_2}\nonumber
\geq\lambda r_1\log\frac{q_1}{r_1}+(1-\lambda)r_2\log\frac{q_2}{r_2}.\nonumber
\end{eqnarray}
Multiplying by $p(y^n||x^n)$ and summing over all $x^n,\ y^n$, and letting $\cal{I}(\textbf{r},\textbf{q})$ be the directed information as in (\ref{calI}), we obtain
\begin{align}
{\cal{I}}(\lambda\textbf{r}_1+(1-\lambda)\textbf{r}_2,\lambda\textbf{q}_1+(1-\lambda)\textbf{q}_2\geq\lambda {\cal{I}}(\textbf{r}_1,\textbf{q}_1)+(1-\lambda){\cal{I}}(\textbf{r}_2,\textbf{q}_2).\nonumber
\end{align}
Further, since the function $\log(x)$ is continuous with continuous partial derivatives, and the directed information is a summation of functions of type $\log(x)$, $\cal{I}(\textbf{r},\textbf{q})$ has the same properties as well. Moreover, it is simple to verify that the sets $A_1,\ A_2$ are both convex, and we can conclude that all conditions in Lemma \ref{lemconv} hold for the directed information.
\end{proof}
\end{lemma}

Recall, that in the alternating maximization procedure we maximize over $\{r(x^n||y^{n-1}),\ q(x^n|y^n)\}$ instead of over $r(x^n||y^{n-1})$ alone, and thus need the following lemma.
\begin{lemma}\label{whyalt}.
For any discrete random variables $X^n,\ Y^n$, the following holds
\begin{align}
\max_{r(x^n||y^{n-1})}I(X^n\rightarrow Y^n)=\max_{r(x^n||y^{n-1}),q(x^n|y^n)}I(X^n\rightarrow Y^n).
\end{align}
The proof will be given after the following supplementary claim, in which we calculate the specific $q(x^n|y^n)$ that maximizes the directed information where $r(x^n||y^{n-1})$ is fixed.
\begin{lemma}\label{rfix}.
For fixed $r(x^n||y^{n-1})$, there exists $c_2(r)=q^{*}(x^n|y^n)$ that achieves $\max_{q(x^n|y^n)}I(X^n\rightarrow Y^n),$ and given by
\begin{align}
q^*(x^n|y^n)=\frac{r(x^n||y^{n-1})p(y^n||x^n)}{\sum_{x^n}{r(x^n||y^{n-1})p(y^n||x^n)}}.\nonumber
\end{align}
\begin{proof}[Proof for Lemma \ref{rfix}]
Let $\textbf{q}^*=q^*(x^n|y^n)$. For any $\textbf{q}=q(x^n|y^n)$, and fixed $\textbf{r}=r(x^n||y^{n-1})$
\begin{align}
&\cal{I}(\textbf{r},\textbf{q}^*)-\cal{I}(\textbf{r},\textbf{q})\nonumber\\
&=\sum_{x^n,y^n}{r(x^n||y^{n-1})p(y^n||x^n)\log{\frac{q^*(x^n|y^n)}{r(x^n||y^{n-1})}}}- \sum_{x^n,y^n}{r(x^n||y^{n-1})p(y^n||x^n)\log{\frac{q(x^n|y^n)}{r(x^n||y^{n-1})}}}\nonumber\\
&=\sum_{x^n,y^n}{r(x^n||y^{n-1})p(y^n||x^n)\log{\frac{q^*(x^n|y^n)}{q(x^n|y^n)}}}\nonumber\\
&=D\left(r(x^n||y^{n-1})p(y^n||x^n)\parallel q(x^n|y^n)\sum_{x^n}{r(x^n||y^{n-1})p(y^n||x^n)}\right)\nonumber\\
&\stackrel{(a)}{\geq} 0,\nonumber
\end{align}
where (a) follows from the non-negativity of the divergence.
\end{proof}
\end{lemma}
\textit{Proof of Lemma \ref{whyalt}.} After finding the PMF $\textbf{q}$ that maximizes $\cal{I}(\textbf{r},\textbf{q})$ where \textbf{r} is fixed, we can see that $q(x^n|y^n)$ is the one that corresponds to the joint distribution  $r(x^n||y^{n-1})p(y^n||x^n)$ in the sense that
\begin{align}
q(x^n|y^n)&=\frac{p(x^n,y^n)}{p(y^n)}\nonumber\\
&=\frac{p(x^n,y^n)}{\sum_{x^n}p(x^n,y^n)}\nonumber\\
&=\frac{r(x^n||y^{n-1})p(y^n||x^n)}{\sum_{x^n}{r(x^n||y^{n-1})p(y^n||x^n)}},\nonumber
\end{align}
and thus, the lemma is proven. \hfill\QED
\end{lemma}

In the following lemma, we find an explicit expression for $\textbf{r}$ that achieves $\max_{r(x^n||y^{n-1})}I(X^n\rightarrow Y^n)$, where $\textbf{q}$ is fixed.
\begin{lemma}\label{qfix}.
For fixed $q(x^n|y^n)$, there exists $c_1(q)=r^{*}(x^n||y^{n-1})$ that achieves $\max_{r(x^n||y^{n-1})}I(X^n\rightarrow Y^n),$ and is given by the products:
$$r^{*}(x^n||y^{n-1})=\prod_{i=1}^n{r(x_i|x^{i-1},y^{i-1})},$$
where
\begin{equation}
r(x_i|x^{i-1},y^{i-1})=\frac{r'(x^i,y^{i-1})}{\sum_{x^i}r'(x^i,y^{i-1})},
\end{equation}
and
\begin{align}
r'(x^i,y^{i-1})=\prod_{x_{i+1}^n,y_{i}^n}{\left[\frac{q(x^n|y^n)}
{\prod_{j=i+1}^{n}{r(x_j|x^{j-1},y^{j-1})}}\right]^{\prod_{j=i}^{n}{p(y_j|x^{j},y^{j-1})}\prod_{j=i+1}^{n}{r(x_j|x^{j-1},y^{j-1})}}}.
\end{align}
\begin{proof}
In order to find the requested \textbf{r}, we find all of its components, namely $\{r(x_i|x^{i-1},y^{i-1})\}_{i=1}^{n}$, by maximizing the directed information over each of them. For convenience, let us use for short: $r_i\triangleq r(x_i|x^{i-1},y^{i-1})$, and $p_i\triangleq p(y_i|x^i,y^{i-1})$.
Since in Lemma \ref{isconcave} we showed that $I(X^n\rightarrow Y^n)$ is concave in $\{\textbf{r},\textbf{q}\}$ and the constraints of the optimization problem are affine, we can use the Lagrange multipliers method with the Karush-Kuhn-Tucker conditions \cite[Ch. 5.3.3]{Boyd}. We define the Lagrangian as:
\begin{align}
J=\sum_{x^n,y^n}{\left(p(y^n||x^n)\prod_{i=1}^{n}{r_i} \log\left(\frac{q(x^n|y^n)}{\prod_{j=1}^{n}{r_j}}\right)\right)}+ \sum_{i=1}^{n}{\left(\sum_{x^{i-1},y^{i-1}}\nu_{i,(x^{i-1},y^{i-1})}\left(\sum_{x_i}{r_i}-1\right)\right)}.\nonumber
\end{align}
Now, for every $i\in\{1,...,n\}$ we find $r_i$ s.t.,
\begin{align}
\frac{\partial J}{\partial r_i}&=\sum_{x_{i+1}^n,y_{i}^n}{\left(p(y^n||x^n)\prod_{j\neq i=1}^{n}{r_j} \left[\log{\frac{q(x^n|y^n)}{\prod_{j=1}^{n}{r_j}}}-1\right]\right)}+\nu_{i,(x^{i-1},y^{i-1})}\nonumber\\
&=\prod_{j=1}^{i-1}{r_{j}}\sum_{x_{i+1}^n,y_{i}^n}{\left(p(y^n||x^n)\prod_{j=i+1}^{n}{r_{j}} \left[\log{\frac{q(x^n|y^n)}{\prod_{j=i+1}^{n}{r_j}}}-\log{\prod_{j=1}^{i-1}{r_{j}}}-\log{r_i}-1\right]\right)}+\nu_{i,(x^{i-1},y^{i-1})}\nonumber\\
&=0.\nonumber
\end{align}
Note that since $\nu_i$ is a function of $(x^{i-1},y^{i-1})$ we can divide the whole equation by $\prod_{j=1}^{i-1}{r_j}$, and get a new $\nu^{*}_{i,(x^{i-1},y^{i-1})}$.\\
Moreover, we can see that three of the expressions in the sum, i.e., $\{\log{\prod_{j=1}^{i-1}{r_{j}}},\ \log{r_i},\  1\}$, do not depend on $(x_{i+1}^n,y_{i}^n)$, thus leaving their coefficient in the equation to be
\begin{align}
\sum_{x_{i+1}^n,y_{i}^n}\left[p(y^n||x^n)\prod_{j=i+1}^{n}{r(x_j|x^{j-1},y^{j-1})}\right]=\prod_{j=1}^{i-1}{p(y_j|x^j,y^{j-1})}.\nonumber
\end{align}
Hence we obtain:
\begin{align}
\log\left[\prod_{x_{i+1}^n,y_{i}^n}{\left(\frac{q(x^n|y^n)}{\prod_{j=i+1}^{n}{r_j}}\right)^{\frac{p(y^n||x^n)\prod_{j=i+1}^{n}{r_{j}}}{{\prod_{j=1}^{i-1}{p_j}}}}}\right]
-\log{r_i}-\log\nu^{**}_{i,(x^{i-1},y^{i-1})}=0,\nonumber
\end{align}
where $$\log{\nu^{**}_{i,(x^{i-1},y^{i-1})}}={\prod_{j=1}^{i-1}{p_j}}\left(1+\log{\prod_{j=1}^{i-1}{r_{j}}}\right)-\nu^{*}_{i,(x^{i-1},y^{i-1})}.$$ Finally, we are left with the expression:
\begin{equation}
r(x_i|x^{i-1},y^{i-1})=\frac{r'(x^i,y^{i-1})}{\sum_{x^i}r'(x^i,y^{i-1})},\nonumber
\end{equation}
where
\begin{align}
r'(x^i,y^{i-1})&=\prod_{x_{i+1}^n,y_{i}^n}{\left[\frac{q(x^n|y^n)}
{\prod_{j=i+1}^{n}{r(x_j|x^{j-1},y^{j-1})}}\right]^{\frac{p(y^n||x^n)\prod_{j=i+1}^{n}{r(x_j|x^{j-1},y^{j-1})}}{\prod_{j=1}^{i-1}{p(y_j|x^{j},y^{j-1})}}}}\nonumber\\
&=\prod_{x_{i+1}^n,y_{i}^n}{\left[\frac{q(x^n|y^n)}
{\prod_{j=i+1}^{n}{r(x_j|x^{j-1},y^{j-1})}}\right]^{\prod_{j=i}^{n}{p(y_j|x^{j},y^{j-1})}\prod_{j=i+1}^{n}{r(x_j|x^{j-1},y^{j-1})}}}.
\end{align}
We can see that for every $i$, $r_i$ depends on $q(x^n|y^n)$ and $\{r_{i+1},r_{i+2},...,r_{n}\}$, and $r_n$ is a function of $q(x^n|y^n)$ alone. Therefore, we can place $r_n$ in the function we have for $r_{n-1}$, thus making $r_{n-1}$ depend on $q(x^n|y^n)$ alone as well. Now we do the same for $r_{n-2}$ and so on until for all $i$, $r_i$ is dependent on $q(x^n|y^n)$ alone. We name this method \textit{Backwards maximization}.
Finally, we obtain $r(x^n||y^{n-1})=\prod_{i=1}^{n}{r_i}$ that maximizes the directed information where $q(x^n|y^n)$ is fixed, i.e., $c_1(q)$, and the lemma is proven.
\end{proof}
\end{lemma}

Having Lemmas \ref{isconcave}-\ref{qfix} we can now state and prove our main theorem.
\begin{theorem}\label{Thc}.
For a fixed channel $p(y^n||x^n)$, there exists an alternating maximization procedure, such as
$I_L$ in Alg. \ref{algc},  to compute
\begin{align}
C_n=\frac{1}{n}\max_{p(x^n||y^{n-1})}I(X^n\rightarrow Y^n).\nonumber
\end{align}
\begin{proof}
To prove Theorem \ref{Thc}, we first have to show existence of a double maximization problem, i.e., an equivalent problem where we maximize over two variables instead of one, and this was shown in Lemma \ref{whyalt}. Now, in order for the alternating maximization procedure to work on this optimization problem, we need to show that the conditions given in Lemma \ref{lemconv} hold here, and this was shown in Lemma \ref{isconcave}, \ref{rfix} and \ref{qfix}. Thus, we have an algorithm for calculating
\begin{align}
C_n=\frac{1}{n}\max_{r(x^n||y^{n-1})}I(X^n\rightarrow Y^n)\nonumber
\end{align}
that is equal to $\lim_{k\rightarrow\infty}I_L(k)$, where $I_L(k)$ is the value of $I_L$ in the $k$th iteration as in Alg. \ref{algc}. Hence, the theorem is proven.
\end{proof}
\end{theorem}

Our last step in proving the convergence of Alg. \ref{algc} is to show why $I_U$ is a tight upper bound. For that reason we state the following theorem.
\begin{theorem}\label{Thupbound}.
For the value of $C_n=\frac{1}{n}\max_{p(x^n||y^{n-1})}I(X^n\rightarrow Y^n)$, the inequality
\begin{align}\label{upboundeq}
C_n\leq I_U,
\end{align}
where
\begin{align}
I_U=\frac{1}{n}\min_r\max_{x_1}\sum_{y_1}\max_{x_2}\cdots\max_{x_n}\sum_{y_{n}}p(y^n||x^n)\log\frac{p(y^n||x^n)}{\sum_{x'^n}p(y^n||x'^n)\cdot
r(x'^n||y^{n-1})}\nonumber
\end{align}
holds. Furthermore, if $r(x^n||y^{n-1})$ achieves $C_n$, then we have
equality in (\ref{upboundeq}).
\end{theorem}
The proof is given in Appendix \ref{proofThUp} for the general case of delay $d$. We also omit the proof of the upper bound for the case where the feedback is a deterministic function of the delayed output, as described in Appendix \ref{appfd}.
\end{section}

\begin{section}{Numerical examples for calculating feedback channel's capacities} \label{secEx}
In this section we present some examples of Alg. \ref{algc} performances over various channels. We start with a memoryless channel to see whether feedback improves the capacity of such channels, and continue with specific FSCs such as the Trapdoor channel and the Ising channel.
Since Alg. \ref{algc} is applicable on Finite State Channels (FSC), we describe this class of such channels and their properties. Gallager \cite{Gal} defined the FSC as one in which the influence of the previous input and output sequence, up to a given point, may be summarized using a \textit{state} with finite cardinality. The FSC is stationary and characterized by the conditional PMF $p(y_i,s_i|x_i,s_{i-1})$ that satisfies
\begin{align}
p(y_i,s_i|x^i,y^{i-1},s^{i-1})=p(y_i,s_i|x_i,s_{i-1}),\nonumber
\end{align}
and the initial state $p(s_0)$.

The causal conditioning probability of the output given the input is given by
\begin{align}
p(y^n||x^n,s_0)=\sum_{s^n}\prod_{i=1}^{n}p(y_i,s_i|x_i,s_{i-1}),\nonumber
\end{align}
and
\begin{align}
p(y^n||x^n)=\sum_{s_0}p(y^n||x^n,s_0)p(s_0).\nonumber
\end{align}
Note that a memoryless channel, i.e., the output at any given time is dependent on the input at that time alone, is an FSC with one state.

It was shown in \cite{PermuterWeissmanGoldsmith} that the capacity of an FSC with feedback is bounded between
\begin{align}
\underline{C}_N-\frac{\log|\mathcal{S}|}{N}\leq C_N\leq \overline{C}_N+\frac{\log|\mathcal{S}|}{N},\label{bound}
\end{align}
where
\begin{align}
\overline{C}=\frac{1}{N}\max_{p(x^n||y^{n-1})}\max_{s_0}I(X^n\rightarrow Y^n|s_0),\\
\underline{C}=\frac{1}{N}\max_{p(x^n||y^{n-1})}\min_{s_0}I(X^n\rightarrow Y^n|s_0).
\end{align}
If we require that the probability of error tends to zero for every initial state $s_0$, then
\begin{align}
C=\lim_{n\rightarrow\infty}\underline{C}.\nonumber
\end{align}
Since these bounds are obtained via maximization of the directed information, we can calculate them using Alg. \ref{algc} as presented in Section \ref{secDescription}, thus estimating the capacity.

Our first example shows the convergence of Alg. \ref{algc} to the analytical capacity of a memoryless channel.
\begin{subsection}{Binary Symmetric Channel}
Consider a memoryless BSC with probability of $p=0.3$ as in Fig. \ref{BSC}.
\begin{figure}[h!]{
  \psfrag{q}[][][0.8]{$\ \;0.3$}\psfrag{p}[][][0.8]{$0.7$}  \psfrag{X}[][][1]{$X$}\psfrag{Y}[][][1]{$Y$}
  \psfrag{0}[][][0.8]{$0$}  \psfrag{1}[][][0.8]{$1$}
 \centerline{ \includegraphics[width=4cm]{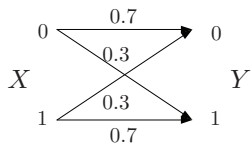}}
 \caption{Binary Symmetric Channel}
 \label{BSC}
}\end{figure} The capacity of this BSC is known to be $C=1-H(0.3)=0.1187$. In Fig. 2 we present
the directed information upper $I_U$ and lower $I_L$ bounds as a function of the iteration (as
given in Alg. \ref{algc}) and compare it to the capacity that is known analytically. Shannon
showed \cite{Shannon2} that for memoryless channels, feedback does not increase the capacity.
Thus, we can expect the numerical solution given in Alg. \ref{algc} to achieve the same value as
in the no-feedback case.
\begin{figure}[h!]{
  \begin{center}
  \psfrag{L}[][][0.8]{\ \ \ \ \ \ \ \ \ \ \ \ \ \ \ \ {\color{blue} Lower bound, $I_L$}}\psfrag{U}[][][0.8]{\ \ \ \ \ \ \ \ \ \ \ \ \ \ \ \ {\color{green} Upper bound, $I_U$}}
  \psfrag{T}[][][0.8]{\ \ \ \ \ \ \ \ \ \  {\color{red} True capacity}} \psfrag{N}[][][1]{Iteration}\psfrag{v}[][][1]{Value}
  \centerline{\includegraphics[width=6cm]{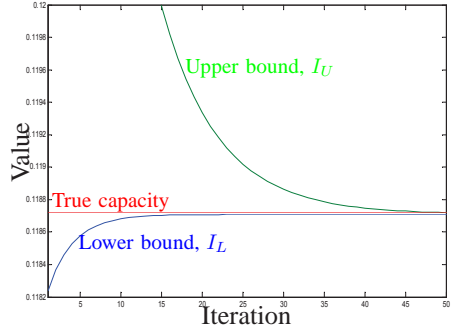}}\ \ \ \ \ \ \ \
  \caption{Performance of Alg. \ref{algc} over BSC(0.3).
  The lower and upper lines are the bounds in each iteration in Alg. \ref{algc}, whereas the horizontal line is the analytical calculation of the capacity.}
  \end{center}
  \label{Bingra}
}\end{figure}
Indeed, we can see that as the iterations number increases, the algorithm approaches the true value and converges. Furthermore, the causally conditioned probability $r(x^n||y^{n-1})$ that Alg. \ref{algc} achieves is actually $r(x^n)$, i.e., does not depend on the feedback. We note here that we can achieve the capacity of the channel using a uniform distribution or $r(x^n)$. This does not imply that there is only one optimum distribution, and indeed the one that Alg. \ref{algc} achieves is not uniform.
\end{subsection}

\begin{subsection}{Trapdoor Channel}
\begin{subsubsection}{Trapdoor channel with 2 states}
The trapdoor channel was introduced by David Blackwell in 1961 \cite{BlackwellTDC} and later on by Ash \cite{Ash}.
\begin{figure}[h!]{
  \psfrag{A}[][][1]{Input}\psfrag{B}[][][1]{Channel}
  \psfrag{C}[][][1]{Output}
 \centerline{ \includegraphics[width=6cm]{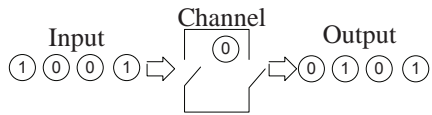}}
 \caption{Trapdoor Channel \cite{Ash}}
}\end{figure}
One can look at this channel as such: Consider a binary channel modulated by a box that contains a single bit referred to as the state. In every step, an input bit is fed to the channel, which then transmits either that bit or the one already contained in the box, each with probability $\frac{1}{2}$. The bit that was not transmitted remains in the box for future steps as the state of the channel.
The state, thus, is the bit in the box, and since it can be '0' or '1', we conclude that $|\mathcal{S}|=2$, or $\log|\mathcal{S}|=1$.

In order to use Alg. \ref{algc}, we first have to calculate the channel probability $p(y^n||x^n,s_0)$. For that purpose, we find $p(y_i|x^i,y^{i-1},s_0)$ analytically.
Note that $p(y_i|x^i,y^{i-1},s_0)=p(y_i|x_i,s_{i-1})$. Thus, first we find the deterministic function for $s_{i-1}$ given the past input, output, and initial state, i.e., $(x^{i-1},y^{i-1},s_0)$, and then the function for $p(y_i|x^i,y^{i-1},s_0)=p(y_i|x_i,s_{i-1})$.
\begin{table}[h]
\begin{minipage}[b]{0.5\linewidth}
\caption{$s_{i-i}$ as a function of $x_{i-1}$, $s_{i-2}$ and $y_{i-1}$}
\centering
\begin{tabular}{|c |c |c|| c|}
\hline
$x_{i-1}$ & $s_{i-2}$ & $y_{i-1}$ &  $s_{i-1}$ \\ [0.5ex]
\hline
0 & 0 & 0 & 0 \\
0 & 0 & 1 & $\phi$ \\
0 & 1 & 0 & 1 \\
0 & 1 & 1 & 0 \\
1 & 0 & 0 & 1 \\
1 & 0 & 1 & 0 \\
1 & 1 & 0 & $\phi$ \\
1 & 1 & 1 & 1 \\
\hline
\end{tabular}
\label{TT1} 
\end{minipage}
\begin{minipage}[b]{0.5\linewidth}
\caption{$p(y_i|s_{i-1},x_i)$}
\centering
\begin{tabular}{|c| c| c|| c|}
\hline
$x_i$ & $s_{i-1}$ & $y_i$ &  $p(y_i|x^i,y^{i-1},s_0=0)$ \\ [0.5ex]
\hline
0 & 0 & 0 & 1 \\
0 & 0 & 1 & 0 \\
0 & 1 & 0 & 0.5 \\
0 & 1 & 1 & 0.5 \\
1 & 0 & 0 & 0.5 \\
1 & 0 & 1 & 0.5 \\
1 & 1 & 0 & 0 \\
1 & 1 & 1 & 1 \\
\hline
\end{tabular}
\label{TT2}
\end{minipage}
\end{table}
An examination of the truth table in Table \ref{TT1} yields the formula for $s_{i-1}$ as
\begin{eqnarray}
s_{i-1}&=&x_{i-1}\oplus y_{i-1}\oplus s_{i-2}\nonumber\\
&=&\bigoplus_{m=1}^{m=i-1}(x_m\oplus y_m)\oplus s_0.\nonumber
\end{eqnarray}
Note that in Table \ref{TT1}, the input series $(0,0,1)$ and $(1,1,0)$ are not possible since the output is not one of the bits in the box; thus we may assign to $s_{i-1}$ whatever value we choose, in order to simplify the formula.
As for the conditional probability $p(y_i|x^i,y^{i-1},s_0)$, we assume that $s_0=0$, and because of the channel's symmetry the outcome for $s_0=1$ is easily calculated. Looking at Table \ref{TT2}, we can see that the formula for $p(y_i|x^i,y^{i-1},s_0=0)$ is given by
\begin{align}
p(y_i|x^i,y^{i-1},s_0=0)=\frac{1}{2}(x_i\oplus s_{i-1})+\overline{(x_i\oplus s_{i-1})}\wedge\overline{(x_i\oplus y_i)},\nonumber
\end{align}
where we know that $s_{i-1}$ is a function of $(x^{i-1},\ y^{i-1},\ s_o)$, and $\wedge$ denotes AND.

Now that we have $p(y^n||x^n,s_0=0)$, we use Alg. \ref{algc} for estimating the capacity of the channel as we run the algorithm to find the upper and lower bound for every $n\in\{1..12\}$, where
\begin{eqnarray}
\overline{C}_n=\max_{s_0}\max_{r(x^n||y^{n-1})}\frac{1}{n}I(X^n\rightarrow Y^n|s_0)+\frac{1}{n},\label{cupp}\\
\underline{C}_n=\max_{r(x^n||y^{n-1})}\min_{s_0}\frac{1}{n}I(X^n\rightarrow Y^n|s_0)-\frac{1}{n}.\label{clow}
\end{eqnarray}
Note that (\ref{cupp}) is calculated via Alg. \ref{algc} and $s_0=0$ due to the channel's symmetry. However, calculating (\ref{clow}) is more difficult, since we have to maximize over all the probabilities $r(x^n||y^{n-1})$, and at the same time minimize over the initial state. Hence, we use another lower bound denoted by $\underline{C}^{*}$, for which $r(x^n||y^{n-1})$ is fixed and is the one that achieves the maximum at (\ref{cupp}), and we only minimize over $s_0$. Clearly, $\underline{C}^{*}\leq\underline{C}$. Fig. \ref{TDgraph} presents the capacity estimation, and the upper and lower bound, as a function of the block length $n$.
\begin{figure}[h!]{
\psfrag{A}[][][0.8]{$\overline{C}_n$}  \psfrag{B}[][][0.8]{$C_n$} \psfrag{C}[][][0.8]{$\underline{C}^*_n$} \psfrag{N}[][][1]{$n$}\psfrag{v}[][][1]{Value}\psfrag{D}[][][0.8]{ True cap.}
 \centerline{ \includegraphics[width=6cm]{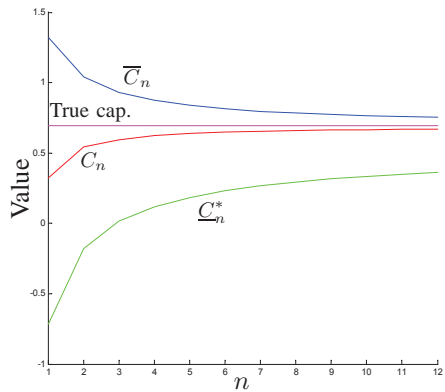}}
  \caption{Plot of $\overline{C}_n,\ C_n,\ \underline{C}^{*}_n$ and the true capacity of the trapdoor channel with 2 states and feedback with delay 1.}
 \label{TDgraph}
}\end{figure}
In \cite{PermuterPaulBenTsachy}, the capacity of the trapdoor channel is calculated analytically, and given by
\begin{align}
C=\lim_{n\rightarrow\infty}C_n=\log\left(\frac{1+\sqrt{5}}{2}\right)\approx 0.69424191.\label{CapTD}
\end{align}
We see from the simulation that the upper and lower bounds of the capacity approach the limit in (\ref{CapTD}), and the estimated capacity at block length $n=12$ is $C_{12}=0.6706533$.
\end{subsubsection}
\begin{subsubsection}{Directed information rate as a different estimator for the capacity}
We now consider an estimator to the feedback capacity of an FSC by calculating
$(n+1)C_{n+1}-nC_n$. The justification for this estimator is based on the following lemma.
\begin{lemma}.
If $\lim_{n\rightarrow\infty}I(X^n;Y_n|Y^{n-1})$ exists, then
\begin{align}
\lim_{n\rightarrow\infty}\frac{1}{n}I(X^n\rightarrow Y^n)=\lim_{n\rightarrow\infty}\left(I(X^n\rightarrow Y^n)-I(X^{n-1}\rightarrow Y^{n-1})\right),\nonumber
\end{align}
i.e.,
\begin{align}
\lim_{n\rightarrow\infty}C_n=\lim_{n\rightarrow\infty}(n+1)C_{n+1}-nC_n.\nonumber
\end{align}
\begin{proof}
If we suppose that the limit above exists, then
\begin{align}
\lim_{n\rightarrow\infty}\left(I(X^n\rightarrow Y^n)-I(X^{n-1}\rightarrow Y^{n-1})\right) &=\lim_{n\rightarrow\infty}I(X^n;Y_n|Y^{n-1})\nonumber\\
&\stackrel{(a)}{=}\lim_{n\rightarrow\infty}\frac{1}{n}\sum_{i=1}^nI(X^i;Y_i|Y^{i-1})\nonumber\\
&=\lim_{n\rightarrow\infty}\frac{1}{n}I(X^n\rightarrow Y^n),\nonumber
\end{align}
where (a) follows from the fact that if the limit of the sequence $\{a_n\}$ exists, then the average of the sequence converges to the same limit.
Further, a result from \cite{Kramer1} provides that if the joint process $\{X_i,Y_i\}$ is stationary, then the limit $\lim_{n\rightarrow\infty}I(X^n;Y_n|Y^{n-1})$ exists.
\end{proof}
\end{lemma}

Fig. \ref{TDdiffgraph} presents the directed information rate estimator using the lemma above, and its comparison to the true capacity.
\begin{figure}[h!]{
\psfrag{N}[][][1]{$n$}\psfrag{v}[][][1]{Value}\psfrag{A}[][][0.8]{$\ \ \ \ \ \ \ \ \ \ \ \ \ \ \ \ \ \ \ \ \ (n+1)C_{n+1}-nC_n$}\psfrag{B}[][][0.8]{\ \ \ \ \  True cap.}
 \centerline{ \includegraphics[width=5cm]{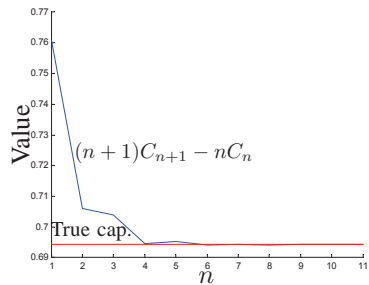}}
 \caption{The upper line is $(n+1)C_{n+1}-nC_n$ calculated using Alg. \ref{algc} and the horizontal line is the analytical calculation, for the trapdoor channel with 2 states and feedback with delay 1.}
 \label{TDdiffgraph}
}\end{figure}
One can see that the convergence of $(n+1)C_{n+1}-nC_n$ is faster than $C_n$ and the upper and lower bounds as seen in Fig. \ref{TDgraph}, and achieves the value $0.6942285$ when we calculate the $11^{th}$ difference. Furthermore, the convergence of the directed information rate stabilizes faster.
\end{subsubsection}
\begin{subsubsection}{M-State Trapdoor channel}
We generalize the trapdoor channel to an M-state one.
\begin{figure}[h!]{
  \psfrag{A}[][][1]{Input}\psfrag{B}[][][1]{Channel}
  \psfrag{C}[][][1]{Output}\psfrag{D}[][][1]{$\cdots$}\psfrag{N}[][][1]{m cells}
 \centerline{ \includegraphics[width=6cm]{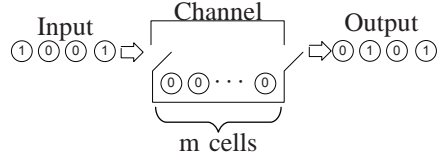}}
 \caption{Trapdoor channel with M states.}
 \label{MstateTD}
}\end{figure}
In the previous example we had $M=2$ cells in the box, one for the state bit, and one for the input bit. One can consider the state to be the number of '1's in the channel before a new input is inserted. We can expand this notation, by letting the 'box' contain more than 2 cells as presented in Fig. \ref{MstateTD}. Here, the state at any given time will express the number of $1'$s that are in the box at that time, and each cell has even probability to be chosen for the output. In this case, $M$ cells in the box are equivalent to $M$ states of the channel.
By that definition we can see that the state $s_{i-1}$ as a function of past input, output, and the initial state is given by
\begin{eqnarray}
s_{i-1}&=&x_{i-1}+s_{i-2}-y_{i-1}\nonumber\\
&=&s_0+\sum_{j=1}^{i-1}(x_j-y_j).\nonumber
\end{eqnarray}
Moreover, for calculating the channel probability $p(y_i=1|x^i,y^{i-1},s_0)$, we add $s_{i-1}$ to $x_i$ and divide the sum by the number of cells, i.e.,
\begin{align}
p(y_i=1|x^i,y^{i-1},s_0)=\frac{s_{i-1}+x_i}{m}.\nonumber
\end{align}
Now that we have $p(y^n||x^n,s_0)$, we use Alg. \ref{algc} for calculating  $C_n$ for every
$n\in\{1,2,...,12\}$. Fig. \ref{TD3graph} presents the directed information rate estimator
$(n+1)C_{n+1}-nC_n$ for the trapdoor channel with $M=3$ cells.
\begin{figure}[h!]{
\psfrag{N}[][][1]{$n$}\psfrag{V}[][][1]{Value}
\centerline{\includegraphics[width=5cm]{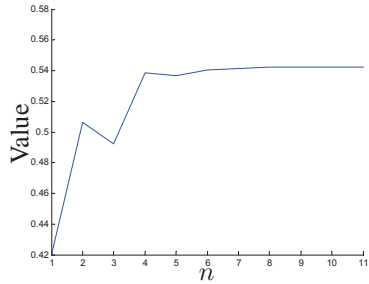}}
\caption{Plot of $(n+1)C_{n+1}-nC_n$ for the trap door channel with 3 cells and feedback with delay 1.}
 \label{TD3graph}
}\end{figure} Note, that in Fig. \ref{TD3graph} we achieve the value $0.5423984$ in the $11^{th}$
difference, thus we can assume that the capacity of a 3-state trapdoor channel is approximately
$0.542$.
\end{subsubsection}
\begin{subsubsection}{Influence of the number of cells on the capacity}
To summarize the trapdoor channel example, we examine the way the number of cells affects the capacity.
\begin{figure}[h!]{
\psfrag{C}[][][0.8]{$C_{12}$} \psfrag{D}[][][0.8]{$\ \ \ \ \ \ \ \ 12C_{12}-11C_{11}$}
\psfrag{m}[][][1]{Cells}\psfrag{v}[][][1]{Value}
 \centerline{ \includegraphics[width=5cm]{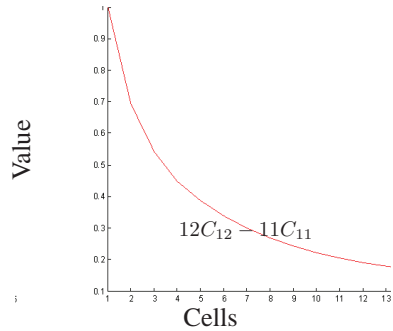}}
  \caption{Change of $12C_{12}-11C_{11}$ over the number of cells in the trapdoor channel with feedback with delay 1.}
 \label{TDcelldelay}
}\end{figure} The estimation use is the directed information rate, with $n=12$. In Fig.
\ref{TDcelldelay} we can see that the capacity decreases as the number of cells increases and
approaches zero. 
\end{subsubsection}
\end{subsection}
\begin{subsection}{The Ising channel}
The Ising model is a mathematical model of ferromagnetism in statistical mechanics. It was originally proposed by the physicist Wilhelm Lenz who gave it as a problem to his student Ernst Ising after whom it is named. The model consists of discrete variables called spins that can be in one of two states. The spins are arranged in a lattice or graph, and each spin interacts only with its nearest neighbors.

The Ising channel is based on its physical model, and simulates Intersymbol Interference where the state of the channel at time $i$ is the current input, and the output is determined by the input at time $i+1$.
\begin{figure}[h!]{
\psfrag{A}[][][1]{$x_n=0$} \psfrag{B}[][][1]{$x_n=1$}  \psfrag{C}[][][1]{$x_{n+1}$} \psfrag{D}[][][1]{$y_n$} \psfrag{h}[][][0.7]{$\frac{1}{2}$}
 \centerline{ \includegraphics[width=8cm]{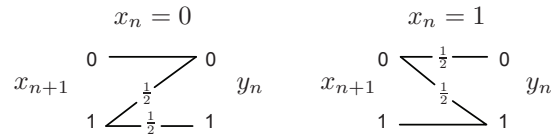}}
  \caption{ The Ising Channel.\cite{BerFla}}
 \label{Isingchannel}
}\end{figure}
The channel (without feedback) was introduced by Berger and Bonomi \cite{BerFla}
and is depicted in Fig. \ref{Isingchannel}. In their paper, they proved the existence of bounds
for the no-feedback case. In addition, they showed that the zero-error capacity without feedback
is $0.5$.

\begin{subsubsection}{Ising channel with delay $d=2$}
We estimate the capacity of the Ising channel with feedback.
\begin{figure}[h!]{
  \begin{center}
  \psfrag{A}[][][0.8]{$\overline{C}_n$}  \psfrag{B}[][][0.8]{$C_n$} \psfrag{C}[][][0.8]{$\underline{C}^{*}_n$} \psfrag{N}[][][1]{$n$}
  \psfrag{v}[][][0.8]{Value}
  \subfloat[]{\includegraphics[width=5cm]{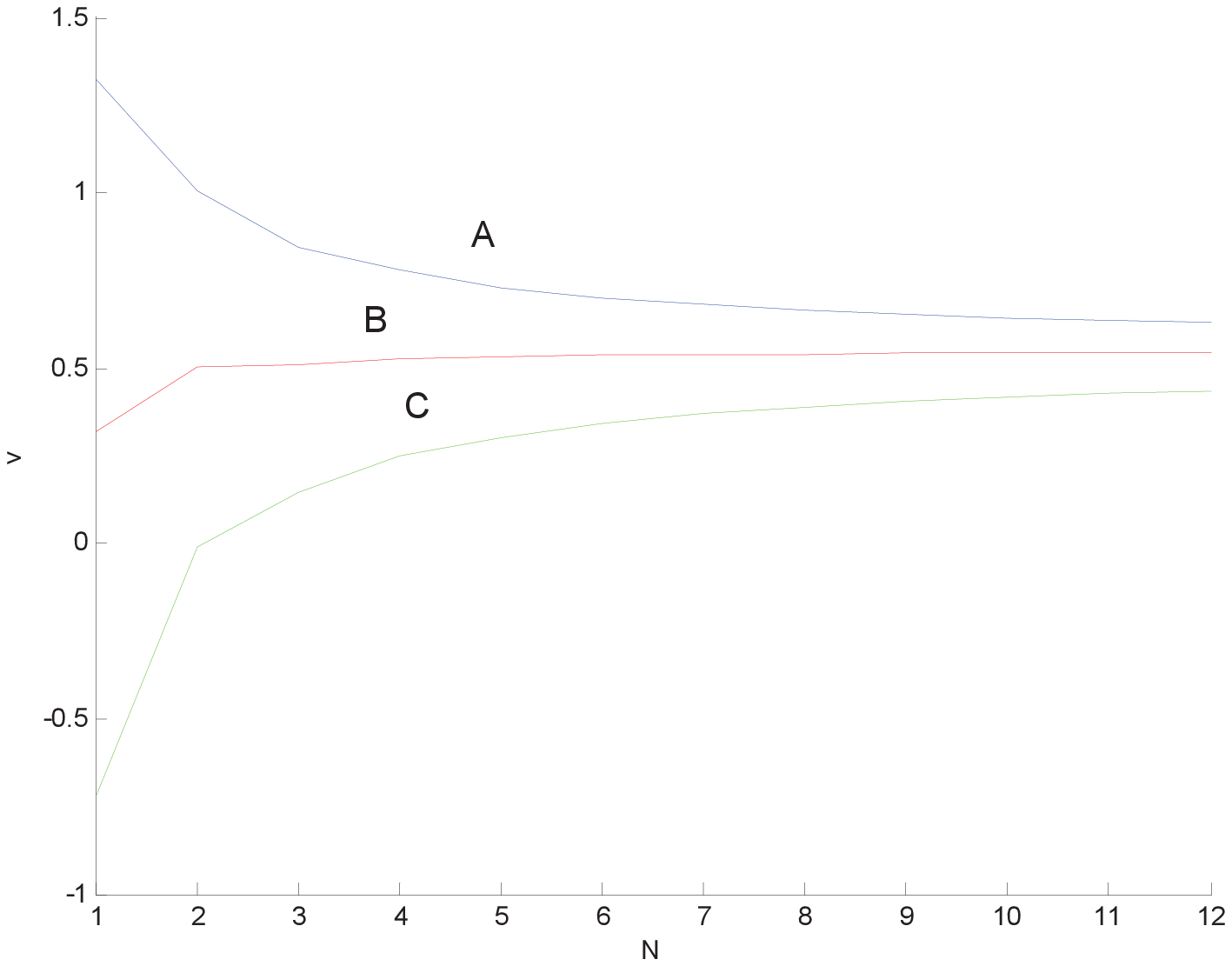}}\ \ \ \ \ \ \ \
  \psfrag{n}[][][1]{$n$}  \psfrag{v}[][][0.8]{Value}
  \subfloat[]{\includegraphics[width=5cm]{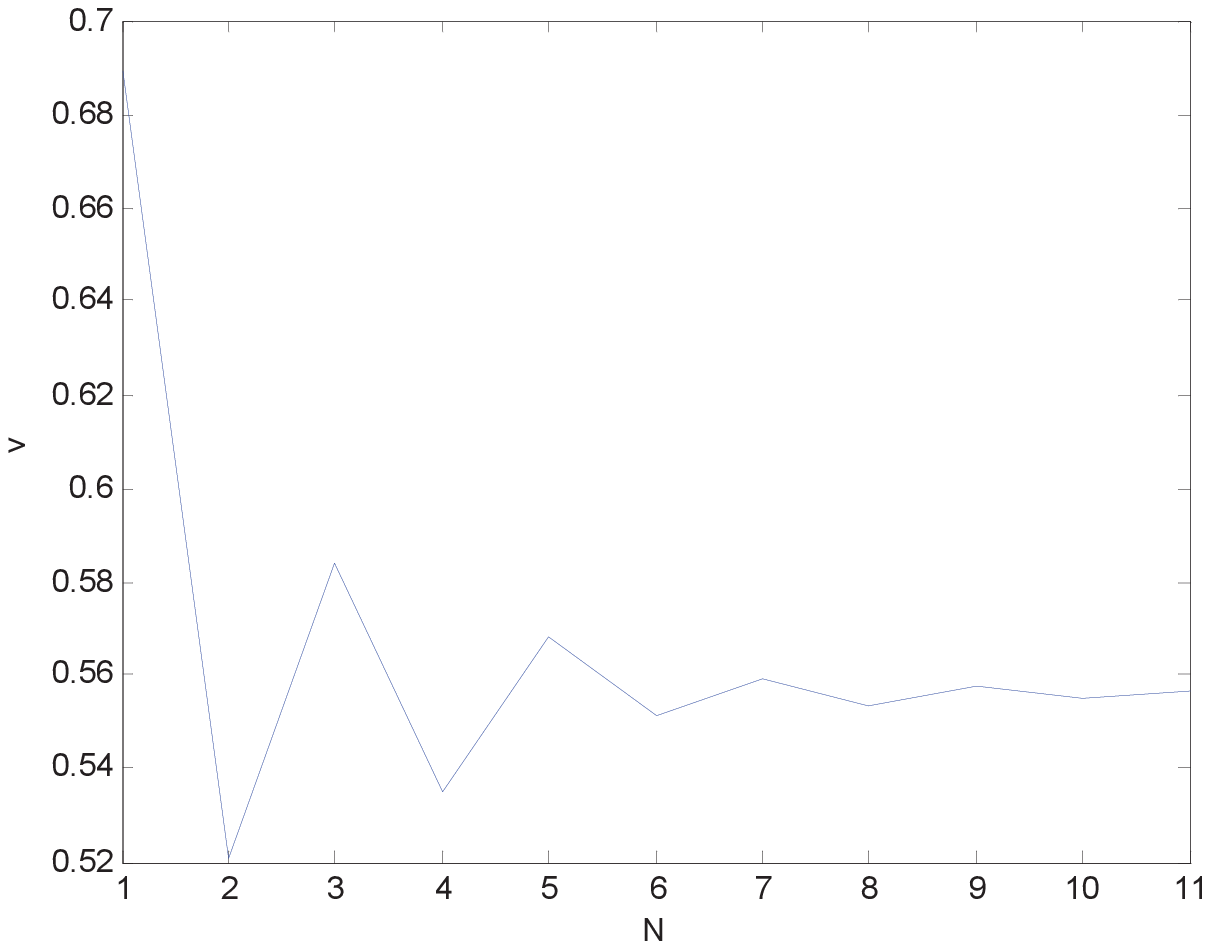}}
  \caption{Performance of Alg. \ref{algc} on the Ising channel with feedback delay of $d=2$. In (a) we present $\overline{C}_n,\ C_n,\ \underline{C}^{*}_n$, and in (b) we have $(n+1)C_{n+1}-nC_n$.}
    \label{IsingGraph2}
 \end{center}
}\end{figure} Since the output at time $i$ is determined by the input at times $i,\ i+1$, we
define the channel PMF as $p(y_0^{n-1}||x^n,s_0)$. Therefore, the feedback at time $i$ must be
the output at time $i-2$, since we cannot have $y_{i-1}$ before $x_{i-1}$ is sent. Thus, looking
at the Ising channel with delay $d=1$ is not a practical example, and we did not examine it. We
ran our algorithm on the Ising channel, with delayed feedback of $d=2$; the results are presented
in Fig. \ref{IsingGraph2}. In Fig. \ref{IsingGraph2} (a), we obtain $C_{12}=0.5459$, and in (b)
we achieve $12C_{12}-11C_{11}=0.5563$ in the $11^{th}$ difference.
\end{subsubsection}
\begin{subsubsection}{The effects the delay has on the capacity}
Here we investigate how the delay influences the capacity. We do so by computing the directed information rate estimator of the Ising channel with blocks of length $12$, over the feedback delay $d=\{2,3,...,12\}$. The formulas for estimating the capacity when the delay is bigger than 1 is given in Section \ref{secDescription}, equations (\ref{Rform1}), (\ref{Qform1}).
In Fig. \ref{Isingdelay} we can see that, as expected, the capacity decreases as the delay increases. This is due to the fact that we have less knowledge of the output to use.
\begin{figure}[h!]{
\psfrag{C}[][][0.8]{$C_{12}$} \psfrag{D}[][][0.8]{$12C_{12}-11C_{11}$} \psfrag{d}[][][1]{Delay of feedback}\psfrag{v}[][][1]{Value}
 \centerline{ \includegraphics[width=5cm]{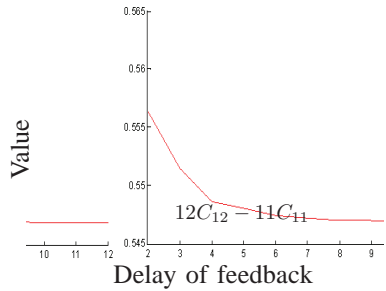}}
  \caption{Change of $12C_{12}-11C_{11}$ over the delay of the feedback on the Ising channel.}
 \label{Isingdelay}
}\end{figure}
\end{subsubsection}
\end{subsection}
\end{section}

\begin{section}{CONCLUSIONS}
In this paper, we generalized the classical BAA for maximizing the directed information over causal conditioning, i.e., calculating
\begin{align}
C_n=\frac{1}{n}\max_{p(x^n||y^{n-1})}I(X^n\rightarrow Y^n).\nonumber
\end{align}
The optimizing the directed information  is necessary for estimating the capacity of an FSC with
feedback. As we attempted to solve this problem we found that difficulties arose regarding the
causal conditioning probability we tried to optimize over. We overcame this barrier by using an
additional backwards loop to find all components of the causal conditioned probability,
separately.

Another application of optimizing the directed information is to estimate the rate distortion function for source coding with feed forward as presented in \cite{WeisMer}, \cite{VenPra}, \cite{VenPra2}. In our future work \cite{NaissPermuter2}, we address the source coding with feedforward problem, and derive bounds for stationary and ergodic sources. We also present and prove a BA-type algorithm for obtaining a numerical solution that computes these bounds.
\end{section}
\appendices
\section{General case for channel coding-Feedback that is a function of the delayed output}\label{appfd}
Here we extend Alg. \ref{algc}, given in Section \ref{secDerivate}, for channels where the encoder has specific information about the delayed output. In this case, the input probability is given by $r(x^n||z^{n-d})$, where $z_i=f(y_i)$ is the feedback, and $f$ is deterministic. In other words, we solve the optimization problem given by
\begin{align}
\max_{r(x^n||z^{n-d})}I(X^n\rightarrow Y^n).\nonumber
\end{align}
The optimization problem is associated to Fig. \ref{dch}.
\begin{figure}[h]{
  \psfrag{x}[][][0.8]{$X^n(M,Z^{i-d})$}\psfrag{C}[][][1]{$p(y^n||x^n)$}\psfrag{y}[][][0.8]{$Y^n$}\psfrag{M}[][][1]{$M$}\psfrag{D}[][][0.8]{Decoder}
  \psfrag{W}[][][1]{$\ \ \ \hat{M}(Y^n)$}  \psfrag{d}[][][1]{$z^{i-d}=f(y^{i-d})$}\psfrag{e}[][][0.8]{Encoder}
 \centerline{ \includegraphics[width=14cm]{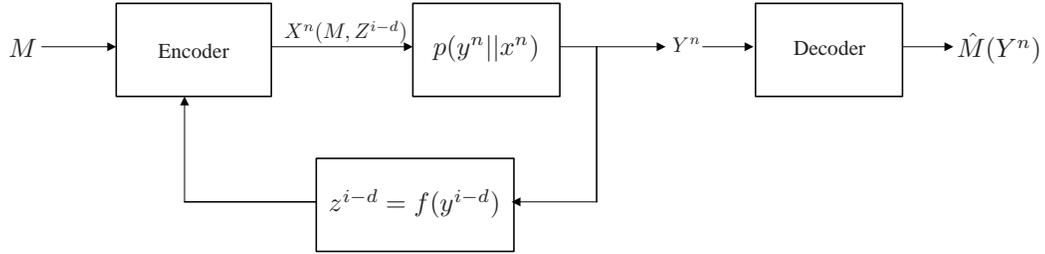}}
 \caption{Channel with delayed feedback as a function of the output.}
 \label{dch}
}\end{figure}

The proof for this case is similar to that of Theorem \ref{Thc}, except the steps that follow from Lemmas \ref{rfix} and \ref{qfix}. Lemma \ref{rfix} proves the existence of an argument $q(x^n|y^n)$ that maximizes the directed information where $r(x^n||y^{n-1})$ is fixed. The modification of this lemma is presented here, where we find the argument $q(x^n|y^n)$ that maximizes the directed information where $r(x^n||z^{n-d})$ is fixed; the proof is omitted. Therefore, the maximization over $q(x^n|y^n)$ where $r(x^n||z^{n-d})$ is fixed is given by
\begin{align}
q^*(x^n|y^n)=\frac{r(x^n||z^{n-d})p(y^n||x^n)}{\sum_{x^n}{r(x^n||z^{n-d})p(y^n||x^n)}}.\nonumber
\end{align}

Lemma \ref{qfix} proves the existence of an argument $r(x^n||y^{n-1})$ that maximizes the directed information where $q(x^n|y^n)$ is fixed. We replace this lemma by Lemma \ref{qfixd}.
\begin{lemma}\label{qfixd}.
For fixed $q(x^n|y^n)$, there exists $c_1(q)$ that achieves $\max_{r(x^n||z^{n-d})}I(X^n\rightarrow Y^n),$ and given by
\begin{align}
r(x^n||z^{n-d})=\prod_{i=1}^n{r(x_i|x^{i-1},z^{i-d})},\nonumber
\end{align}
where
\begin{equation}
r(x_i|x^{i-1},z^{i-d})=\frac{r'(x^i,z^{i-d})}{\sum_{x^i}r'(x^i,z^{i-d})},
\end{equation}
and
\begin{align}
r'(x^i,z^{i-d})=\prod_{x_{i+1}^n,y_{i-d+1}^n}\prod_{A_{i,d,z}}{\left[\frac{q(x^n|y^n)}
{\prod_{j=i+1}^{n}{r(x_j|x^{j-1},z^{j-d})}}\right]^{\frac{p(y^n||x^n)\prod_{j=i+1}^{n}{r(x_j|x^{j-1},z^{j-d})}}{\sum_{A_{i,d,z}}\prod_{j=1}^{i-d}{p(y_j|x^{j},y^{j-1})}}}}.
\end{align}
\begin{proof}
We find the products of $r(x^n||z^{n-d})$ that achieve maximum for the directed information.
For convenience, let us use for short: $r_i\triangleq r(x_i|x^{i-1},z^{i-d})$, and $p_i\triangleq p(y_i|x^i,y^{i-1})$. As in Lemma \ref{isconcave} we can omit that $I(X^n\rightarrow Y^n)$ is concave in $\{r(x^n||y^{n-d}),\ q(x^n|y^n)\}$. Furthermore, the constraints of the optimization problem are affine, and we can use the Lagrange multipliers method with the Karush-Kuhn-Tucker conditions. We define the Lagrangian as:
\begin{align}
J=\sum_{x^n,y^n}{\left(p(y^n||x^n)\prod_{i=1}^{n}{r_i} \log\left(\frac{q(x^n|y^n)}{\prod_{j=1}^{n}{r_j}}\right)\right)}+ \sum_{i=1}^{n}{\left(\sum_{x^{i-1},z^{i-d}}\nu_{i,(x^{i-1},z^{i-d})}\left(\sum_{x_i}{r_i}-1\right)\right)}.\nonumber
\end{align}
Now, for every $i\in\{1..n\}$ we find $r_i$ s.t.,
\begin{align}
\frac{\partial J}{\partial r_i}&=\sum_{x_{i+1}^n,y_{i-d+1}^n,A_{i,d,z}}{\left(p(y^n||x^n)\prod_{j\neq i=1}^{n}{r_j} \left[\log{\frac{q(x^n|y^n)}{\prod_{j=1}^{n}{r_j}}}-1\right]\right)}+\nu_{i,(x^{i-1},z^{i-d})}\nonumber\\
&=\sum_{A_{i,d,z}}\prod_{j=1}^{i-1}{r_{j}}\sum_{x_{i+1}^n,y_{i-d+1}^n}{\left(p(y^n||x^n)\prod_{j=i+1}^{n}{r_{j}} \left[\log{\frac{q(x^n|y^n)}{\prod_{j=i+1}^{n}{r_j}}}-\log{\prod_{j=1}^{i-1}{r_{j}}}-\log{r_i}-1\right]\right)}+\nu_{i,(x^{i-1},z^{i-d})}\nonumber\\
&=0,\nonumber
\end{align}
where the set $A_{i,d,z}=\{y^{i-d}:z^{i-d}=f(y^{i-d})\}$ stands for all output sequences $y^{i-d}$ s.t. the function in the delay maps them to the same sequence $z^{i-d}$, which is the feedback.

Note that since $\prod_{j=1}^{i-1}{r_j}$ does not depend on $A_{i,d,z}$, we can take the product out of the sum. Furthermore,
since $\nu_i$ is a function of $(x^{i-1},z^{i-d})$ we can divide the whole equation by the product above, and get a new $\nu^{*}_{i,(x^{i-1},z^{i-d})}$.
Moreover, we can see that three of the expressions in the sum, i.e., $\{\log{\prod_{j=1}^{i-1}{r_{j}}},\ \log{r_i},\  1\}$, do not depend on $(x_{i+1}^n,y_{i-d+1}^n)$, thus leaving their coefficient in the equation to be
\begin{align}
\sum_{x_{i+1}^n,y_{i-d+1}^n,A_{i,d,z}}{p(y^n||x^n)\prod_{j=i+1}^{n}{r_{j}}}=\sum_{A_{i,d,z}}\prod_{j=1}^{i-d}{p_j}.\nonumber
\end{align}
Hence we obtain:
\begin{align}
\log\left[\prod_{x_{i+1}^n,y_{i-d+1}^n}{\left(\frac{q(x^n|y^n)}{\prod_{j=i+1}^{n}{r_j}}\right)^{\frac{p(y^n||x^n)\prod_{j=i+1}^{n}{r_{j}}}{\sum_{A_{i,d,z}}{\prod_{j=1}^{i-d}{p_j}}}}}\right]
-\log{r_i}-\log\nu^{**}_{i,(x^{i-1},z^{i-d})}=0,\nonumber
\end{align}
where $$\log{\nu^{**}_{i,(x^{i-1},z^{i-d})}}=\sum_{A_{i,d,z}}{\prod_{j=1}^{i-d}{p_j}}\left(1+\log{\prod_{j=1}^{i-1}{r_{j}}}\right)-\nu^{*}_{i,(x^{i-1},z^{i-d})}.$$ Therefore, we are left with the expression:
\begin{equation}
r(x_i|x^{i-1},z^{i-d})=\frac{r'(x^i,z^{i-d})}{\sum_{x^i}r'(x^i,z^{i-d})},\nonumber
\end{equation}
where
\begin{align}
r'(x^i,z^{i-d})=\prod_{x_{i+1}^n,y_{i-d+1}^n,A_{i,d,z}}{\left[\frac{q(x^n|y^n)}
{\prod_{j=i+1}^{n}{r(x_j|x^{j-1},z^{j-d})}}\right]^{\frac{p(y^n||x^n)\prod_{j=i+1}^{n}{r(x_j|x^{j-1},z^{j-d})}}{\sum_{A_{i,d,z}}\prod_{j=1}^{i-d}{p(y_j|x^{j},y^{j-1})}}}}.
\end{align}
As in Section \ref{secDerivate}, we can see that for all $i$, $r_i$ is dependent on $q(x^n|y^n)$ and $\{r_{i+1},r_{i+2},...,r_{n}\}$, and $r_n$ is a function of $q(x^n|y^n)$ alone. Thus, we use the \textit{Backwards maximization} method.
After calculating $r_i$ for all $i={1,...,n}$, we obtain $r(x^n||z^{n-d})=\prod_{i=1}^{n}{r_i}$ that maximizes the directed information where $q(x^n|y^n)$ is fixed, i.e., $c_1(q)$ and the lemma is proven.
\end{proof}
\end{lemma}

As mentioned, by replacing Lemmas \ref{rfix}, \ref{qfix} by those given here, we can follow the outline of Theorem \ref{Thc} and conclude the existence of an alternating maximization procedure, i.e., we can compute
\begin{align}
C_n=\frac{1}{n}\max_{r(x^n||z^{n-d})}I(X^n\rightarrow Y^n)\nonumber
\end{align}
that is equal to $\lim_{k\rightarrow\infty}I_L(k)$, where $I_L(k)$ is the value of $I_L$ in the
$k$th iteration in the extended algorithm. One more step is required in order to prove the
extension of Alg. \ref{algc} to the case presented here; the existence of $I_U$. This part is
presented in Appendix \ref{proofThUp}.

\section{Proof of Theorem \ref{Thupbound}}\label{proofThUp}
Here, we prove the existence of an upper bound, $I_U$, that converges to $C_n$ from above
simultaneously with the convergence on $I_L$ to it from below, as in Alg. \ref{algc}. To this
purpose, we present and prove few  lemmas that assist in obtaining our main goal. We start with
Lemma \ref{Upbound1} that gives an inequality for the directed information. This inequality is
used in Lemma \ref{Upbound2} to prove the existence of our upper bound which Lemma
\ref{Upbound2tight} proves to be tight. Theorem \ref{Thupbound} combines Lemmas \ref{Upbound2},
\ref{Upbound2tight}.
\begin{lemma}\label{Upbound1}.
Let $I_{r_1}(X^n\rightarrow Y^n)$ correspond to $r_1(x^n||y^{n-d})$, then for every $r_0(x^n||y^{n-d})$,
\begin{align}
I_{r_1}(X^n\rightarrow Y^n)\leq\sum_{x^n,y^{n-d}}r_1(x^n||y^{n-d})\sum_{y_{n-d+1}^n}p(y^n||x^n)\log\frac{p(y^n||x^n)}{\sum_{x'^n}p(y^n||x'^n)\cdot r_0(x'^n||y^{n-d})}.\nonumber
\end{align}
\begin{proof}
For any $r_1(x^n||y^{n-d})$, $r_0(x^n||y^{n-d})$,
\begin{align}
\sum_{x^n,y^{n-d}}r_1(x^n||y^{n-d})&\sum_{y_{n-d+1}^n}p(y^n||x^n)\log\frac{p(y^n||x^n)}{\sum_{x'^n}p(y^n||x'^n)\cdot r_0(x'^n||y^{n-d})}-I_{r_1}(X^n\rightarrow Y^n)\nonumber\\
&=\sum_{x^n,y^n}r_1(x^n||y^{n-d})\cdot p(y^n||x^n)\log{\frac{p(y^n||x^n)}{\sum_{{x'}^n}p(y^n||x'^n)\cdot r_0(x'^n||y^{n-d})}}\nonumber\\
&\ \ \ -\sum_{x^n,y^n}r_1(x^n||y^{n-d})\cdot p(y^n||x^n)\log{\frac{p(y^n||x^n)}{\sum_{{x'}^n}p(y^n||x'^n)\cdot r_1(x'^n||y^{n-d})}}\nonumber\\
&=\sum_{x^n,y^n}r_1(x^n||y^{n-d})\cdot p(y^n||x^n)\log{\frac{\sum_{{x'}^n}p(y^n||x'^n)\cdot r_1(x'^n||y^{n-d})}{\sum_{{x'}^n}p(y^n||x'^n)\cdot r_0(x'^n||y^{n-d})}}\nonumber\\
&=\sum_{y^n}p_1(y^n)\log{\frac{p_1(y^n)}{p_0(y^n)}}\nonumber\\
&\stackrel{(a)}{=}D\left(p_1(y^n)||p_0(y^n)\right)\nonumber\\
&\stackrel{(b)}{\geq}0,\nonumber
\end{align}
where in (a), $p_0(y^n)$ and $p_1(y^n)$ are the PMFs of $y^n$ that corresponds to $r_0(x'^n||y^{n-d})$ and $r_1(x'^n||y^{n-d})$, and (b) follows from the non negativity of the divergence. Thus, the lemma is proven.
\end{proof}
\end{lemma}

Our next lemma uses the inequality in Lemma \ref{Upbound1} to show the existence of the upper bound, which is the first step in proving Theorem \ref{Thupbound}.
\begin{lemma}\label{Upbound2}.
For every $r_0(x^n||y^{n-d})$,
\begin{align}
C_n\leq I_U,\nonumber
\end{align}
where
\begin{align}
I_U=\frac{1}{n}\max_{x^d}\sum_{y_1}\max_{x_{d+1}}\sum_{y_2}\cdots\max_{x_n}\sum_{y_{n-d+1}^n}p(y^n||x^n)\log\frac{p(y^n||x^n)}{\sum_{x'^n}p(y^n||x'^n)\cdot r_0(x'^n||y^{n-d})}.\nonumber
\end{align}
\begin{proof}
To prove this lemma, we first use lemma \ref{Upbound1}. For every $r_1(x^n||y^{n-d}),\ r_0(x^n||y^{n-d})$,
\begin{align}
I_{r_1}(X^n\rightarrow Y^n)&\stackrel{(a)}{\leq}\sum_{x^n,y^{n-d}}r_1(x^n||y^{n-d})\sum_{y_{n-d+1}^n}p(y^n||x^n)\log\frac{p(y^n||x^n)}{\sum_{x'^n}p(y^n||x'^n)\cdot r_0(x'^n||y^{n-d})}\nonumber\\
&\stackrel{(b)}{\leq}\sum_{x^n,y^{n-d}}\prod_{i=1}^n r_1(x_i|x^{i-1},y^{i-d})\underbrace{\max_{x_n}\sum_{y_{n-d+1}^n}p(y^n||x^n)\log\frac{p(y^n||x^n)}{\sum_{x'^n}p(y^n||x'^n)\cdot r_0(x'^n||y^{n-d})}}_\textrm{$f(x^{n-1},y^{n-d})$}\nonumber\\
&\stackrel{(c)}{=}\sum_{x^{n-1},y^{n-d-1}}\prod_{i=1}^{n-1} r_1(x_i|x^{i-1},y^{i-d})\underbrace{\sum_{y_{n-d}}\max_{x_n}\sum_{y_{n-d+1}^n}p(y^n||x^n)\log\frac{p(y^n||x^n)}{\sum_{x'^n}p(y^n||x'^n)\cdot r_0(x'^n||y^{n-d})}}_\textrm{$f(x^{n-1},y^{n-d-1})$}\nonumber\\
&\leq\sum_{x^{n-1},y^{n-d-1}}\prod_{i=1}^{n-1} r_1(x_i|x^{i-1},y^{i-d})\underbrace{\max_{x_{n-1}}\sum_{y_{n-d}}\max_{x_n}\sum_{y_{n-d+1}^n}p(y^n||x^n)\log\frac{p(y^n||x^n)}{\sum_{x'^n}p(y^n||x'^n)\cdot r_0(x'^n||y^{n-d})}}_\textrm{$f(x^{n-2},y^{n-d-1})$}\nonumber\\
&\vdots\nonumber\\
&\leq\sum_{x^{d}}\prod_{i=1}^d r_1(x_i|x^{i-1},y^{i-d})\underbrace{\sum_{y_1}\max_{x_{d+1}}\sum_{y_2}\cdots\max_{x_n}\sum_{y_{n-d+1}^n}p(y^n||x^n)\log\frac{p(y^n||x^n)}{\sum_{x'^n}p(y^n||x'^n)\cdot r_0(x'^n||y^{n-d})}}_\textrm{$f(x^d)$}\nonumber\\
&\leq\underbrace{\sum_{x^{d}}\prod_{i=1}^d r_1(x_i|x^{i-1},y^{i-d})}_\textrm{$=1$}\underbrace{\max_{x^d}\sum_{y_1}\max_{x_{d+1}}\sum_{y_2}\cdots\max_{x_n}\sum_{y_{n-d+1}^n}p(y^n||x^n)\log\frac{p(y^n||x^n)}{\sum_{x'^n}p(y^n||x'^n)\cdot r_0(x'^n||y^{n-d})}}_\textrm{$\in\mathbb{R}$}\nonumber\\
&=\max_{x^d}\sum_{y_1}\max_{x_{d+1}}\sum_{y_2}\cdots\max_{x_n}\sum_{y_{n-d+1}^n}p(y^n||x^n)\log\frac{p(y^n||x^n)}{\sum_{x'^n}p(y^n||x'^n)\cdot r_0(x'^n||y^{n-d})},\nonumber
\end{align}
where (a) follows Lemma \ref{Upbound1}, (b) follows from maximizing an expression over $x_n$, and (c) follows from the fact that the expression in the under-brace is a function of $x^{n-1},y^{n-d}$, and we can take it out of the summation over $x_n$ and use $\sum_{x_n}r(x_n|x^{n-1},y^{n-d})=1$. The rest of the steps are the same as (b) and (c), where we refer to a different $x_i$.

Since the inequality above is true for every $r_1(x^n||y^{n-d})$, we can use it on $r_c(x^n||y^{n-d})$ that achieves $C_n$, and thus for every $r_0(x^n||y^{n-d})$
\begin{align}
C_n\leq\frac{1}{n}\max_{x^d}\sum_{y_1}\max_{x_{d+1}}\sum_{y_2}\cdots\max_{x_n}\sum_{y_{n-d+1}^n}p(y^n||x^n)\log\frac{p(y^n||x^n)}{\sum_{x'^n}p(y^n||x'^n)\cdot r_0(x'^n||y^{n-d})}.\nonumber
\end{align}
This is also true for every $r_0(x^n||y^{n-d})$, and hence for the minimum over all $r_0(x^n||y^{n-d})$, and we obtain
\begin{align}
C_n\leq\frac{1}{n}\min_{r_0}\max_{x^d}\sum_{y_1}\max_{x_{d+1}}\sum_{y_2}\cdots\max_{x_n}\sum_{y_{n-d+1}^n}p(y^n||x^n)\log\frac{p(y^n||x^n)}{\sum_{x'^n}p(y^n||x'^n)\cdot r_0(x'^n||y^{n-d})},\nonumber
\end{align}
and the lemma is proven.
\end{proof}
\end{lemma}

The next part of Theorem \ref{Thupbound} is to show that the bound is tight.
\begin{lemma}\label{Upbound2tight}.
The upper bound in Lemma \ref{Upbound2} is tight, and is obtained by $r(x^n||y^{n-d})$ that achieves the capacity.
\begin{proof}
In Lemma \ref{Upbound2}, we showed only half of the proof of the theorem, i.e., the existence of an upper bound. To prove this lemma, we need to show that this inequality is tight. For that purpose, we use the Lagrange multipliers method with the KKT conditions with respect to all $r(x_i|x^{i-1},y^{i-d})$s. We can use the KKT conditions since the directed information is a concave function in all $r(x_i|x^{i-1},y^{i-d})$s, as seen in Lemma \ref{whyalt}.

We define the Lagrangian as
\begin{align}
J=\sum_{x^n,y^n}r(x^n||y^{n-d})&\cdot p(y^n||x^n)\log{\frac{p(y^n||x^n)}{\sum_{{x'}^n}p(y^n||x'^n)\cdot r(x'^n||y^{n-d})}}\nonumber\\
&-\sum_{i=1}^n\sum_{x^{i-1},y^{i-d}}\nu_{i,(x^{i-1},y^{i-d})}(\sum_{x_i}r(x_i|x^{i-1},y^{i-d})-1)
+\sum_{i=1}^n\sum_{x^{i-1},y^{i-d}}h_{i,(x^{i-1},y^{i-d})}r(x_i|x^{i-1},y^{i-d}).\nonumber
\end{align}
Now, for every $r(x_i|x^{i-1},y^{i-d})$, we have
\begin{align}
\frac{\partial J}{\partial r(x_i|x^{i-1},y^{i-d})}&=\sum_{x_{i+1},y_{i-d+1}}r(x_{i+1}|x^{i},y^{i-d+1})\cdots\sum_{x_{n},y_{n-d}}r(x_n|x^{n-1},y^{n-d})\cdot\nonumber\\
&\ \ \ \ \ \ \sum_{y_{n-d+1}^n}p(y^n||x^n)\log\frac{p(y^n||x^n)}{\sum_{x'^n}p(y^n||x'^n)\cdot r(x'^n||y^{n-d})}-\nu_{i,(x^{i-1},y^{i-d})}+h_{i,(x^{i-1},y^{i-d})}.\nonumber
\end{align}
Setting $\frac{\partial J}{\partial r(x_i|x^{i-1},y^{i-d})}=0$ we are left with two cases. For $r(x_i|x^{i-1},y^{i-d})>0$ the KKT conditions requires us to set $h_i=0$ and we obtain
\begin{align}
\sum_{x_{i+1},y_{i-d+1}}r(x_{i+1}|x^{i},y^{i-d+1})\cdots\sum_{x_{n},y_{n-d}}r(x_n|x^{n-1},y^{n-d})
\sum_{y_{n-d+1}^n}p(y^n||x^n)\log\frac{p(y^n||x^n)}{\sum_{x'^n}p(y^n||x'^n)\cdot r(x'^n||y^{n-d})}=\nu_i,\nonumber
\end{align}
whereas for $r(x_i|x^{i-1},y^{i-d})=0$ we set $h_i>0$ and the equality becomes an inequality.

We now analyze our results for the case where $r(x_i|x^{i-1},y^{i-d})>0$. First, we note that for $i=n$ we have that
\begin{align}
\sum_{y_{n-d+1}^n}p(y^n||x^n)\log\frac{p(y^n||x^n)}{\sum_{x'^n}p(y^n||x'^n)\cdot r(x'^n||y^{n-d})}=\nu_{n,(x^{n-1},y^{n-d})},\nonumber
\end{align}
and thus constant for every $x_n$. As a result, for $i=n-1$ we have
\begin{align}
\sum_{x_n,y_{n-d}}r(x_n|x^{n-1},y^{n-d})&\sum_{y_{n-d+1}^n}p(y^n||x^n)\log\frac{p(y^n||x^n)}{\sum_{x'^n}p(y^n||x'^n)\cdot r(x'^n||y^{n-d})}\nonumber\\ &=\sum_{y_{n-d}}\max_{x_n}\sum_{y_{n-d+1}^n}p(y^n||x^n)\log\frac{p(y^n||x^n)}{\sum_{x'^n}p(y^n||x'^n)\cdot r(x'^n||y^{n-d})}\nonumber\\
&=\nu_{n-1,(x^{n-2},y^{n-d-1})}\nonumber
\end{align}
that again, is constant for every $x_{n-1}$. We can move backwards and obtain that for $i=1$,
\begin{align}
\sum_{x_2}r(x_2|x_1)&\cdots\sum_{x_d}r(x_d|x^{d-1})\sum_{x_{d+1},y_1}r(x_{d+1}|x^d,y_1)\cdots r(x_n|x^{n-1},y^{n-d})\cdot\nonumber\\
&\ \ \ \ \sum_{y_{n-d+1}^n}p(y^n||x^n)\log\frac{p(y^n||x^n)}{\sum_{x'^n}p(y^n||x'^n)\cdot r(x'^n||y^{n-d})}\nonumber\\
&=\sum_{x_2}r(x_2|x_1)\underbrace{\max_{x_3^d}\sum_{y_1}\max_{x_{d+1}}\cdots\max_{x_n}\sum_{y_{n-d+1}^n}p(y^n||x^n)\log\frac{p(y^n||x^n)}{\sum_{x'^n}p(y^n||x'^n)\cdot r(x'^n||y^{n-d})}}_\textrm{$\nu_{2,(x_1)}$}\nonumber\\
&=\underbrace{\max_{x_2^d}\sum_{y_1}\max_{x_{d+1}}\cdots\max_{x_n}\sum_{y_{n-d+1}^n}p(y^n||x^n)\log\frac{p(y^n||x^n)}{\sum_{x'^n}p(y^n||x'^n)\cdot r(x'^n||y^{n-d})}}_\textrm{$\nu_{1}$}\nonumber\\
&=\max_{x^d}\sum_{y_1}\max_{x_{d+1}}\cdots\max_{x_n}\sum_{y_{n-d+1}^n}p(y^n||x^n)\log\frac{p(y^n||x^n)}{\sum_{x'^n}p(y^n||x'^n)\cdot r(x'^n||y^{n-d})}.\nonumber
\end{align}
Using the analysis above, we find an expression for $C_n$ where $r(x^n||y^{n-d})$ achieves it. In the following equations we can assume that $r(x^n||y^{n-d})>0$, since otherwise, for the specific $x^n,y^n$, the expression for $C_n$ will contribute $0$ to the summation.
\begin{align}
C_n&=\frac{1}{n}\sum_{x^n,y^n}r(x^n||y^{n-d})\cdot p(y^n||x^n)\log\frac{p(y^n||x^n)}{\sum_{x'^n}p(y^n||x'^n)\cdot r(x'^n||y^{n-d})}\nonumber\\
&=\frac{1}{n}\sum_{x_1}r(x_1)\sum_{x_2}r(x_2|x_1)\cdots\sum_{x_d}r(x_d|x^{d-1})\sum_{x_{d+1},y_1}r(x_{d+1}|x^d,y_1)\nonumber\\
&\ \ \ \cdots r(x_n|x^{n-1},y^{n-d})\sum_{y_{n-d+1}^n}p(y^n||x^n)\log\frac{p(y^n||x^n)}{\sum_{x'^n}p(y^n||x'^n)\cdot r(x'^n||y^{n-d})}\nonumber\\
&=\frac{1}{n}\sum_{x_1}r(x_1)\max_{x^d}\sum_{y_1}\max_{x_{d+1}}\cdots\max_{x_n}\sum_{y_{n-d+1}^n}p(y^n||x^n)\log\frac{p(y^n||x^n)}{\sum_{x'^n}p(y^n||x'^n)\cdot r(x'^n||y^{n-d})}\nonumber\\
&\stackrel{(a)}{=}\frac{1}{n}\max_{x^d}\sum_{y_1}\max_{x_{d+1}}\cdots\max_{x_n}\sum_{y_{n-d+1}^n}p(y^n||x^n)\log\frac{p(y^n||x^n)}{\sum_{x'^n}p(y^n||x'^n)\cdot r(x'^n||y^{n-d})},\nonumber
\end{align}
where (a) is due to the analysis above for $i=1$. We showed that the upper bound is tight, and thus the lemma is proven.
\end{proof}
\end{lemma}

Now we combine both lemmas to conclude our main theorem.
\begin{proof}[Proof of Theorem \ref{Thupbound}]
As showed in Lemma \ref{Upbound2}, there exists an upper bound for $C_n$. Lemma \ref{Upbound2tight} showed that this upper bound is tight, when using the PMF $r(x^n||y^{n-d})$ that achieves $C_n$. Thus, the theorem is proven.
\end{proof}

\textit{Generalization of Theorem \ref{Thupbound}} We generalize Theorem \ref{Thupbound} to the case where the feedback is a delayed function of the output (as presented in Appendix \ref{appfd}). We recall, that the optimization problem for this model is
\begin{align}
\max_{r(x^n||z^{n-d})}I(X^n\rightarrow Y^n).\nonumber
\end{align}
While solving this optimization problem, we defined the following set: $A_{i,d,z}=\{y^{i-d}:z^{i-d}=f(y^{i-d})\}$; namely, all output sequences $y^{i-d}$ s.t. the function in the delay sends them to the same sequence $z^{i-d}$. We use this notation for the upper bound.
In that case, the upper bound is of the form
\begin{align}
I_U=\frac{1}{n}\max_{x^d}\sum_{z_1}\max_{x_{d+1}}\cdots\sum_{z_{n-d}}\max_{x_n}\sum_{A_{n,d,z}}\sum_{y_{n-d+1}^n}p(y^n||x^n)\log\frac{p(y^n||x^n)}{\sum_{x'^n}p(y^n||x'^n)\cdot r(x'^n||z^{n-d})}.\nonumber
\end{align}
The proof for this upper bound is omitted due to its similarity to the case where $z_i=y_i$ for
all $i$, i.e., Theorem \ref{Thupbound}. Moreover, one can see that this is a generalization,
since if indeed $z_i=y_i$, then $A_{n,d,z}$ has only one sequence, $y_{n-d}$, and the equation
for $I_U$ coincides with the one in Theorem \ref{Thupbound}.

\end{document}